\documentclass{ectj}

\usepackage{amsfonts,amssymb,graphics,epsfig,verbatim,bm,latexsym,amsmath,url,amsbsy}
\graphicspath{{./}}
\DeclareGraphicsExtensions{.png,.eps}

\newtheorem{assumption}{Assumption}
\newtheorem{proposition}{Proposition}

\newtheorem{lemma}{Lemma}
\newtheorem{definition}{Definition}
\newtheorem{observation}{Observation}

\makeatletter
\@addtoreset{theorem}{section}
\@addtoreset{assumption}{section}
\@addtoreset{proposition}{section}
\@addtoreset{corollary}{section}
\@addtoreset{lemma}{section}
\@addtoreset{definition}{section}
\@addtoreset{observation}{section}
\@addtoreset{algorithm}{section}
\@addtoreset{example}{section}
\@addtoreset{remark}{section}
\makeatother

\title[Persistence is not independence]{Persistence is not Independence: a Structural Credibility Protocol for Distribution based Recovery in VARLiNGAM.}

\author[T.~A.~Doan Nguyen and A.~Bucciol]{Tu~Anh~Doan~Nguyen$^{\dagger}$ and Alessandro~Bucciol$^{\dagger}$}

\address{$^{\dagger}$Department of Economics, University of Verona,
            Via Cantarane 24, 37129 Verona, Italy.}
\email{nguyentuanh.doan@studenti.univr.it, alessandro.bucciol@univr.it}

\newcommand{\OAmap}[1]{\ifx#1A B\else\ifx#1B C\else\ifx#1C D\else\ifx#1D E\else\ifx#1E F\else\ifx#1F G\fi\fi\fi\fi\fi\fi}
\newcommand{\OA}[1]{Appendix~\OAmap{#1}}
\newcommand{\MT}[2]{#2}
\makeatletter
\def\@submissionheader{Preprint}
\makeatother

\begin{document}

\begin{abstract}
VAR--LiNGAM recovers a contemporaneous structure from reduced-form innovations under mutual
independence of the underlying structural shocks. Recovery itself and bootstrap
stability of the recovered edges are natural post-estimation outputs, but neither directly
assesses the identifying independence assumption. We organize post-recovery validation as an ordered,
non-compensatory protocol with four layers (reduced-form admissibility, distributional
informativeness, independence support and edge persistence) and establish, by a
constructive population result, a Monte Carlo experiment under known ground truth and a
firm-level Compustat application, that the last two layers are non-nested: an edge can persist
under resampling while the recovered shocks remain dependent, and independence can be supported
while the leading recovered edge fails to persist. Independence support is classified operationally by a
block-permutation dHSIC diagnostic on the recovered shocks, with mutual information as a
continuous measure of residual dependence and tail co-exceedance as a complementary
extreme-dependence diagnostic. Empirically, the conjunction of no detected independence violation and a reportable
persistent leading edge occurs only in a small subset.
\end{abstract}

\section{Introduction}
\label{sec:intro}
{\sloppy Reduced-form VAR innovations do not carry structural meaning on their own; see \citet{sims1980}
and \citet{lutkepohl2005}. Structural VARs supply the missing identifying logic through external
restrictions (\citealp{blanchardquah1989}; \citealp{uhlig2005}; \citealp{kilianlutkepohl2017});
distribution-based identification supplies it instead through the shape of the shock distribution
(\citealp{comon1994}; \citealp{hyvarinen2000}). VAR--LiNGAM is an established causal-discovery
instance of the latter route in economics: under non-Gaussianity, mutual independence of the
structural shocks and an acyclic contemporaneous structure, the contemporaneous decomposition is
identified up to permutation and scaling (\citealp{shimizu2006}; \citealp{hyvarinen2010};
\citealp{moneta2013}).\par}

Identification theory says when a structural object is recoverable under maintained assumptions. This paper
asks what empirical evidence entitles the recovered object to retain its structural interpretation once
recovery has taken place. Because the algorithm returns a contemporaneous matrix, an ordering or a graph directly from the
residuals, identification can look as though it comes for free. Recovery itself and the bootstrap
stability of its edges are the natural post-estimation outputs to report in support of a
recovered object. Neither reaches the assumption that licenses a structural
reading. Within the maintained non-Gaussian acyclic route, \emph{independence} of the underlying structural
shocks is the condition that supplies structural separation; after recovery, the compatibility of the
recovered shocks with independence is the empirical object that can be assessed. Non-Gaussianity keeps
the distribution-based route open, acyclicity restricts admissible form, and the graph is a downstream
representation. Independence support and bootstrap persistence are non-nested evidentiary objects.
Distribution-based recovery does not make structural entitlement automatic.

Existing work provides tests and inference procedures for independence and related distributional
identifying assumptions in ICA and non-Gaussian SVARs, including methods that account for
estimated shocks, serial dependence and weak non-Gaussianity
(\citealp{wu2009}; \citealp{amengual2022}; \citealp{amengual2024}; \citealp{drautzburg2023};
\citealp{herwartz2018}; \citealp{montielolea2022}; \citealp{hoesch2024}), and the LiNGAM
software package reports both error-independence $p$-values and bootstrap edge frequencies
\citep{ikeuchi2023}.
Our contribution is not a new general test of independence. We instead study the post-recovery
relation between independence support and bootstrap persistence, show that the two are
non-nested, and operationalize their joint reporting in short, heavy-tailed VAR--LiNGAM
applications. That both diagnostics are available side by side while, to our knowledge, their
relationship has not been formalized as two non-nested post-recovery evidentiary objects, is the
motivation for the paper.

Our first contribution is to formalize the post-recovery evidentiary problem. We cast the step after the
algorithm returns an output as a \emph{structural credibility protocol}: an ordered, non-compensatory sequence of four
evidentiary layers, namely reduced-form admissibility, distributional informativeness,
independence support and edge persistence, in which independence support is decisive and a
favourable result in one layer cannot compensate for weak support in another; the architecture determines what
structural interpretation a recovered object is entitled to carry. One implication of the protocol is a directional-standing
standard for what a surviving orientation must establish before it is read as directed, with or without structural
entitlement (Section~\ref{sec:standing}).

Our second contribution is to establish that \emph{independence support and bootstrap persistence are
non-nested}. Persistence concerns sampling variability of a recovered feature; independence support concerns the
identifying assumption. We establish it at three levels. As a population statement
it is a constructive existence result: within a class of stationary structural VARs there are
processes in every cell of the support-by-persistence table, because shock dependence and edge
strength are separate features of the data-generating process, one direction obtained under
consistency and bootstrap-regularity conditions for the correctly specified independent-shock
case and the other under an explicit pseudo-true recovery condition
(Appendix~\ref{app:nn}, Proposition~\ref{prop:constr}); for the shared-scale construction used, that pseudo-true
population recovery is derived exactly rather than assumed (Proposition~\ref{prop:sym}). A Monte Carlo experiment
under known ground truth recovers both
off-diagonal cells through the full recovery-and-diagnosis pipeline, with approximately nominal size of the
independence-support diagnostic (Section~\ref{sec:mc}). Empirically, the two dimensions come apart
across sixty admissible Compustat firms. Persistence therefore cannot substitute for independence
support, a point that matters because bootstrap stability is a common robustness check in applied
work. To make independence support a jointly reported quantity alongside persistence, we classify it operationally with a
block-permutation Hilbert--Schmidt Independence Criterion (HSIC) diagnostic
(\citealp{gretton2005}; \citealp{pfister2018}), read alongside recovered-shock mutual information
(MI), a continuous measure of residual dependence, and a tail co-exceedance diagnostic
\citep{coles2001}.

The Monte Carlo also exposes an operational point of independent interest: the magnitude cutoff
used to count an edge as present in a bootstrap resample must not be small relative to bootstrap
sampling noise. With the small cutoff of our baseline implementation, a non-existent edge is declared persistent in
roughly half of the replications at $T=250$; in the independent-shock design, a cutoff of about
two bootstrap standard deviations sharply reduces this false persistence while retaining strong
edges, at the cost of lower sensitivity to weak edges.

We apply the protocol to firm-level Compustat data, one of the original economic application
domains of VAR--LiNGAM \citep{moneta2013}, with the Tennessee Eastman Process (TEP) of
\citet{downsvogel1993} as a long-sample comparative benchmark. Persistence and independence support
diverge across units: of sixty admissible Compustat systems, eleven have a persistent leading
edge without independence support, twenty have independence support without a persistent leading
edge, nine have both and twenty have neither. The
combined admissibility and distributional-informativeness screens retain only $60$ of $11{,}945$
candidate firms, and the
supported-and-persistent set is small, with heterogeneous direction across firms.

The paper is not a critique of VAR--LiNGAM as an identification device; under its assumptions the
identification result is correct, and the protocol is not a replacement identification theorem. Its formal
result concerns the evidentiary relationship among objects observed after a distribution-based identification
procedure has been applied: entitlement after recovery. We treat the
block-permutation HSIC as a serial-aware diagnostic rather than an exact test on recovered shocks
and flag the exact-inference question in Section~\ref{sec:disc}.

\section{Related literature}
\label{sec:lit}
\subsection{Structural VAR identification}
Since \citet{sims1980}, the central difficulty is that reduced-form innovations do not determine a
unique contemporaneous structure: $\Sigma_u=C\Sigma_eC'$ leaves the impact matrix $C$
underidentified. Classical solutions impose external restrictions, such as recursive zeros,
long-run neutrality \citep{blanchardquah1989} or sign restrictions \citep{uhlig2005}; see
\citet{lutkepohl2005} and \citet{kilianlutkepohl2017}.

\subsection{Distribution-based identification}
Independent component analysis shows that a linear mixture of non-Gaussian, mutually independent
sources is identified up to permutation and scaling (\citealp{comon1994}; \citealp{hyvarinen2000}).
\citet{shimizu2006} turn this into LiNGAM for acyclic linear systems, \citet{hyvarinen2010} extend
it to VAR residuals and \citet{moneta2013} bring it into applied economics. That study places
both distributional validation and bootstrap robustness inside the VAR--LiNGAM workflow: its
algorithm explicitly recommends validating non-Gaussianity and ``(at least approximate)
statistical independence'' of the estimated components before proceeding, and its empirical
applications assess robustness of the recovered structure through bootstrap resampling,
subsample comparisons and alternative estimation methods.
Our point is not that these ingredients were absent. We separate them as distinct post-recovery
evidentiary objects and study their relationship directly: independence support concerns
compatibility with the identifying assumption, persistence concerns reproducibility of the
recovered relation under resampling. In this sense the protocol is complementary to the original
VAR--LiNGAM program: it does not alter the identification argument, but organizes the evidence
once a structure has been recovered. The ingredients are established; the contribution lies in
the evidentiary architecture. In parallel, an
econometric literature studies non-Gaussian SVAR identification directly: identification and
(quasi-)maximum likelihood estimation (\citealp{lanne2017}; \citealp{gourieroux2017};
\citealp{fiorentini2023}), GMM based on higher moments (\citealp{lanneluoto2021};
\citealp{keweloh2021}; \citealp{guay2021}) and the survey of \citet{lewis2025}.
Recent work by \citet{lanneliuluoto2026} establishes identification results for non-Gaussian SVARs under
particular forms of shock dependence. Our question is different: holding the VAR--LiNGAM identification route
fixed, we examine the post-recovery evidentiary relationship between support for its independence restriction and
the persistence of recovered relations.
\citet{montielolea2022} question whether higher-moment identification solves the
simultaneous-causality problem; \citet{monetapallante2022} evaluate ICA-based SVAR identification by
simulation.

\subsection{Testing independence of structural shocks}
Closest to our Layer~3 are formal tests of the independence assumption on estimated shocks.
\citet{amengual2022} derive moment tests of independent components and \citet{amengual2024}
specification tests for non-Gaussian SVARs; \citet{herwartz2018} identifies shocks by minimizing a
dependence criterion and \citet{herwartz2019} uses independence to revisit long-run restrictions;
\citet{drautzburg2023} use independence to refine set identification; \citet{mattesontsay2017}
build ICA on distance covariance; \citet{wu2009} propose a smoothed-bootstrap mutual-information
test for strictly stationary multivariate series and apply it to estimated ICA components, and
\citet{amengual2022,amengual2024} account explicitly for the sampling variability induced by
estimating the shocks. \citet{montielolea2022} recommend testing the stronger shock assumptions
that higher-moment identification imposes; \citet{hoesch2024} develop inference that remains
valid when identifying non-Gaussianity is weak, treating the strength of the identifying
distributional assumptions as an empirical question; and \citet{amengual2025} extend
independence-based specification testing to Gaussian SVARs identified by conventional
restrictions. These contributions treat independence as part of estimation or as a
specification test of the structural model. We take recovery as given, measure independence
support on the recovered shocks with a serial-aware reference and, our main point, relate that
support to bootstrap persistence; the non-nestedness of the two is the object of study, not the
independence test itself.

\subsection{Dependence measures, resampling and post-selection inference}
Mutual information is the natural information-theoretic measure \citep{cover2006};
nearest-neighbour estimators are consistent but high-variance in short series. Kernel measures such
as HSIC (\citealp{gretton2005}; \citealp{gretton2008}) and distance covariance \citep{szekely2007}
detect general dependence and admit permutation references; \citet{pfister2018} extend HSIC to
joint independence of several variables. For serially dependent data, kernel independence testing
has its own literature: \citet{chwialkowski2014} show that i.i.d.-style resampling is invalid for
random processes and develop corrected procedures, \citet{wang2021} test independence between
stationary multivariate series with a validated residual bootstrap, \citet{liu2023} extend dHSIC
to joint independence of stationary and non-stationary series, and \citet{shen2024} establish
asymptotic validity of block-permutation independence testing for stationary temporal data. Our
implementation uses block permutation in the spirit of the block bootstrap
(\citealp{kunsch1989}; \citealp{politisromano1994}) and of these temporal extensions; its formal
validity on recovered shocks is a separate question (Section~\ref{sec:disc}). Dependence
concentrated in the joint tails is the domain of extreme-value dependence, where the co-exceedance
coefficient is the standard summary (\citealp{coles2001}; \citealp{embrechts2002}). Finally,
\citet{leeb2005}, \citet{berk2013} and \citet{lee2016} show that treating a data-driven choice as
fixed invalidates the usual guarantees. Distribution-based recovery poses a similar problem one
level up: the object carried forward is an entire contemporaneous structure selected by the shape
of the estimated shock distribution. We do not provide exact post-recovery inference; we place the
question within this tradition.

\section{Framework: recovery, persistence and independence support}
\label{sec:fw}
Let $y_t$ be $n$-dimensional with reduced-form VAR($p$) $y_t=A_1y_{t-1}+\cdots+A_py_{t-p}+u_t$.
VAR--LiNGAM decomposes the innovation as $u_t=Bu_t+e_t$, i.e.\ $(I-B)u_t=e_t$, with $B$ the
contemporaneous coefficient matrix and $e_t$ the structural shocks; equivalently $u_t=Ce_t$ and
$e_t=Wu_t$. Second moments $\Sigma_u=C\Sigma_eC'$ do not pin down $C$. Under non-Gaussian, mutually
independent $e_{it}$ with acyclic $B$, $C$ is identified up to permutation and scaling
(\citealp{comon1994}; \citealp{shimizu2006}; \citealp{hyvarinen2010}). Identification is therefore
\emph{conditional}: it holds if the assumptions hold. The three assumptions do different work.
Non-Gaussianity provides information beyond covariance but does not make components structural;
acyclicity restricts admissible form; independence provides the structural separation. Favourable
evidence on non-Gaussianity or acyclic form cannot compensate for weak independence support, and a
recovered graph is a representation of the output, not evidence for it.

\subsection{Three objects}
\begin{definition}[Recovery]
A decomposition is recovered if the algorithm returns $\hat B$, an ordering or a graph on an
admissible innovation layer.
\end{definition}
\begin{definition}[Persistence]
\label{def:pers}
A recovered edge is persistent if, across moving-block-bootstrap resamples, it recurs, with
absolute recovered coefficient at least $\eta$ and with the sign of the point estimate, at rates
(stability, sign consistency) exceeding thresholds $(\underline s,\underline c)$.
\end{definition}
\begin{definition}[Independence support]\sloppy
A recovered decomposition has independence support if the recovered shocks $\hat e_t$ are
compatible with mutual independence under the designated serial-aware diagnostic.
\end{definition}
Recovery is guaranteed by construction on any admissible input. Persistence concerns sampling
variability. Independence support concerns the identifying assumption.
\begin{observation}[Non-nestedness]
\label{prop:nn}
\textnormal{The set of recovered systems with a persistent leading edge and the set with
independence support are non-nested: neither contains the other. Appendix~\ref{app:nn} constructs
stationary structural VAR processes in each cell of the support-by-persistence table
(Proposition~\ref{prop:constr}), one direction from standard consistency conditions alone and the
other under an explicit pseudo-true recovery condition (Assumption~\ref{a:pseudo});
Section~\ref{sec:mc} recovers the cells under known ground truth;
Section~\ref{sec:nonnest} documents them across sixty Compustat systems.}
\end{observation}

\subsection{The structural credibility protocol}
\emph{Maintained model conditions.} The protocol evaluates structural credibility conditional on the
maintained VAR--LiNGAM model class: linearity, causal sufficiency (no hidden common causes) and an
acyclic contemporaneous representation. Reduced-form dynamic adequacy is not maintained but
assessed, imperfectly, by Layer~1 through lag choice, residual autocorrelation and stationarity
diagnostics. The four layers assess whether
the recovered object receives empirical support \emph{within} that maintained class; they do not
validate the maintained class itself, and ``structural credibility'' is to be read throughout as
structural credibility within the maintained model class.

The support required for a structural reading is organized as an ordered, non-compen\-satory
protocol with four layers (Table~\ref{tab:protocol}). The layers operate at two levels.
Layers~1--3 are cumulative conditions on the system: they determine whether the recovered
decomposition is admissible, distributionally informative and compatible with the identifying
assumption. Layer~4 operates on individual recovered relations: it filters which edges are
reproducible enough to report. A failure at Layer~4 therefore reduces what may be reported, not
the admissibility of the system, whereas a failure at Layer~3 undermines the structural reading
itself; a favourable result in one layer cannot compensate for weak support in another, and
persistence is a robustness filter, not a substitute for independence support.

\begin{table}[t]
\caption{\label{tab:protocol}The four-layer structural credibility protocol.}
\begin{center}
\small
\begin{tabular*}{\textwidth}{@{}l@{\extracolsep{\fill}}p{0.20\textwidth}p{0.32\textwidth}p{0.34\textwidth}@{}}
\hline\hline
 & Layer & Diagnostic object & Interpretive role \\
\hline
1 & Reduced-form admissibility & VAR lag, residual autocorrelation, stationarity & Entry condition for applying LiNGAM to the residual layer \\
2 & Distributional informativeness & Non-Gaussianity tests on residuals & Opens the distribution-based route; does not warrant interpretation \\
3 & Independence support & dHSIC classifier; MI severity measure; tail reinforcement diagnostic & Core evidentiary basis for structural credibility \\
4 & Edge persistence & Moving-block bootstrap stability and sign consistency & Robustness filter for recovered relations \\
\hline\hline
\end{tabular*}
\end{center}
\end{table}

\subsection{Directional standing: what a surviving orientation must establish}
\label{sec:standing}
The protocol says where the structural reading stops. It would be of limited use if it said only that. The
protocol does not rescue $B$ after independence support weakens; it changes the question asked of $B$: not
whether the recovered orientation is causal, but whether its direction carries systematic information that
simpler or reversed representations do not. We therefore separate three things that the layers of
Table~\ref{tab:protocol} might otherwise be read as deciding together. Layers~1--2 determine whether a recovered
directional object is admissible for assessment at all. Layer~3 determines whether it retains structural
entitlement within the maintained model class. Directional standing is a separate burden on the surviving
orientation, assessed on the edges that Layer~4 allows to be reported, and it is assessed whether or not
Layer~3 passes.
\begin{definition}[Directional standing]
\label{def:standing}
On a specification that passes Layers~1--2, a recovered orientation has directional standing if (i) it carries
systematic information that the reverse orientation of the same pairs does not, (ii) it adds information
relative to representations computable from the same data without the identifying assumptions---undirected or
symmetric measures, and reduced-form directional measures---and (iii) it is not an artefact of one recovery
algorithm, each of (i)--(iii) being assessed by a stated statistic against a stated null or reference.
\end{definition}
The operational forms are, respectively, a directional statistic (for example a shock-conditioned
tail-dependence lift) computed forward over the estimated edges and compared with the same pairs reversed and
with matched non-links, the forward-minus-reverse gap being tested against a permutation null with its
attainable floor reported, preferably out of sample from pre-period estimates alone; rank and partial
correlations, given size and given the first principal component, between the directed scores and each
benchmark, a low correlation establishing distinctness and not added information, which requires a task in
which the directed scores improve on the benchmark; and re-estimation under a defensible alternative recovery
algorithm of the same admissibility class, with edge-set overlap and ranking agreement reported and each
route's Layer-3 verdict stated---a change of lag order or pre-whitening within the same algorithm is a
sensitivity check, not an alternative route.

Three readings follow, and only three. If Layer~3 passes and standing holds, the orientation may be read as
directed structure within the maintained model class. If Layer~3 fails and standing holds, the orientation may
be read as a directed transmission pattern---a non-causal reading with directional empirical content---but
not as structure: the Layer-3 failure withdraws structural entitlement, and standing does not restore it. If
standing fails, the recovered network is a description of the output, whatever Layer~3 says. Failing
(i)--(iii) therefore does not invalidate the recovered network as a description; it withdraws the directed
reading, exactly as a Layer-3 failure withdraws the structural one. A paper that makes a directed claim states
which of the three assessments were run, with which thresholds, which failed, and which of the three readings
it claims. The standard is deliberately one of evidence rather than a test: what it requires is that the
burden carried by the surviving orientation be visible.

\section{Assessing independence support}
\label{sec:method}
Layer~3 is decisive for the identification route; the founding study's algorithm already calls
for validating the independence of the estimated components (Section~\ref{sec:lit}), and the
protocol makes that evidentiary role explicit and separates it from robustness. The recovered shocks
$\hat e_t=\hat Wu_t$ are the object on which independence can be assessed. We report three
measures on them, with roles assigned in Section~\ref{sec:read}: the block-permutation dHSIC of
Section~\ref{sec:hsic} is the operational classifier, MI a continuous descriptive measure and tail
co-exceedance a reinforcement diagnostic. The descriptive measure is pairwise recovered-shock MI
among the $\hat e_{it}$, estimated by the $k$-nearest-neighbour estimator of \citet{kraskov2004}
with $k=5$. MI is sensitive to general dependence in population, but its nearest-neighbour
estimator is high-variance in short series and can have limited power against dependence
concentrated in sparsely observed tails; the covariance whitening that precedes ICA is insensitive
to such higher-order dependence altogether. The two measures that follow address these
limitations.

\subsection{Block-permutation HSIC}
\label{sec:hsic}
For a pair $(\hat e_i,\hat e_j)$ with Gaussian-kernel Gram matrices $K,L$ (median-heuristic
bandwidth) and $H=I-T^{-1}\bm 1\bm 1'$, $\widehat{\mathrm{HSIC}}=T^{-2}\,\mathrm{tr}(KHLH)$, which
has population value zero if and only if the variables are independent for characteristic kernels
\citep{gretton2005}. For the system we use the $d$-variable joint statistic
$\widehat{\mathrm{dHSIC}}$ of \citet{pfister2018} over all recovered shocks. Inference is by
\emph{block} permutation: the first series is held fixed and each remaining series is shuffled in
contiguous blocks of length $\ell$, preserving within-block serial dependence while destroying
cross-series dependence under the null (\citealp{kunsch1989}; \citealp{politisromano1994}). The
$p$-value uses the standard finite-permutation correction,
$\hat p=\bigl(1+\sum_{b=1}^{B}\mathbf 1\{T^*_b\ge T_{\mathrm{obs}}\}\bigr)/(B+1)$ for $B$
permutations, so the smallest attainable value is $1/(B+1)$. A free permutation destroys the serial dependence of the recovered shocks and
can understate null variability: on the $14$-variable TEP shocks the free-permutation null standard deviation is
about thirty times smaller than the block one (\OA{B}). Block permutation mitigates
this distortion but is not an exact test on recovered shocks, which carry dynamic-estimation and
structural-recovery error; we quantify its finite-sample size and power by simulation in
Section~\ref{sec:mc}.

\subsection{Tail co-exceedance}
For a high quantile $q$, $\widehat\chi_{ij}(q)=\widehat{\Pr}(\hat e_j>Q_j(q)\mid \hat
e_i>Q_i(q))$, with the lower tail symmetric. Under independence $\widehat\chi_{ij}(q)\to1-q$;
larger values signal dependence surviving in the extremes (\citealp{coles2001};
\citealp{embrechts2002}). We report the maximum pairwise co-exceedance with a block-bootstrap
band. \OA{A} shows by simulation that in series of $n\approx40$ the tail diagnostic
registers tail-aligned dependence that HSIC misses in about a quarter of replications.

\subsection{Reading them together}
\label{sec:read}
The three measures play different roles. The block-permutation dHSIC is the operational
classifier of independence support: a system is declared supported when it does not reject at
level $\alpha$. Throughout, ``independence support'' denotes compatibility with independence
under the stated diagnostic at the chosen level; it does not mean that independence has been
established. MI is a continuous descriptive measure of the severity of the remaining dependence,
used to rank and to compare raw VAR residuals with recovered shocks, not to classify. Tail co-exceedance is a
reinforcement diagnostic aimed at a specific failure mode (dependence confined to the tails) and is
reported alongside the classifier, not folded into it. Agreement across the three is stronger
diagnostic support, not a formal decision rule.

\emph{Large systems.} The joint dHSIC is the protocol's classifier. When an application does not compute a
joint product-kernel statistic over dozens of recovered shocks, the protocol admits a pairwise family, every
pair of recovered shocks tested under the same serial-aware null, with false-discovery control over the full
family, as the operational classifier; the application states the difference in evidential scope, and the
choice and calibration of the diagnostic remain its own. The verdict carries a different label: when the procedure yields no
rejection at its level, the system shows \emph{no detected pairwise violation}, or has \emph{pairwise
independence support}; it is not given the status of support under the joint classifier, because pairwise
independence is necessary but not sufficient for the joint independence that identification requires. Three
reading rules attach. The null must remain serial-aware, because free permutation understates null
variability on recovered shocks (\OA{B}); circular time shifts beyond the serial-dependence horizon of the
squared shocks are the natural analogue of block permutation when the family is large. Decisions are read
jointly with severity measures of effect-size type, such as the average correlation between squared recovered
shocks, which do not grow with the sample; at large $T$ consistent pairwise tests detect small departures,
so rejection shares alone are not informative. And a global combination $p$-value over the family,
or the count under a more conservative correction, is a decision-level summary at another level of
conservativeness, not a severity measure, because it depends on the sample size.

\emph{Every network that is read.} The layers apply to the specification that is interpreted, not to a
baseline alone. A training-window network, a sub-sample network, an alternative-route estimate or an
unfiltered variant carries its own Layer-3 verdict on its own recovered shocks; a variant whose verdict fails
may be compared with the baseline as a sensitivity but is not read as structure. Under the non-compensatory
rule this is not an additional requirement; it is the same requirement applied to each object that is read.

\section{Data and empirical design}
\label{sec:data}
\subsection{Domains}
Compustat firm-level data are one of the original application domains of VAR--LiNGAM in applied
economics \citep{moneta2013}; our extract and variable set follow that tradition but are not a
re-analysis of any published study. The extract spans 1970--2024 (SIC 2000--3999); after cleaning, a redundancy
filter and firm-year aggregation, four variables remain: gross profit (gp), revenue (revt), R\&D
(xrd) and employment (emp). The unit is the firm-level series. TEP is a comparative benchmark
\citep{downsvogel1993}: fault-free data from a single simulation run, filtered to a $14$-variable
system on $499$ differenced observations. It is a methodological benchmark, not a second
application: a long regular series in which the sampling behaviour of the diagnostics can be seen
without the short-sample noise that dominates Compustat. The $14$-variable specification contains a
near-collinear pair of innovations, $\texttt{xmeas\_10}$ and $\texttt{xmv\_6}$, which motivated the
sensitivity check reported in Section~\ref{sec:res} and \OA{F}.

\subsection{Pipeline}
Layer 1: first differencing; a reduced-form VAR (lag 2 fixed for Compustat, for comparability
across short series; AIC lag for TEP); Ljung--Box residual autocorrelation, ADF and KPSS
stationarity. Layer 2: non-Gaussianity screening (D'Agostino--Pearson and Jarque--Bera) of the
residual layer; identification requires only that at most one structural shock be Gaussian, so
this screen, which requires every residual series to reject Gaussianity, is a conservative
operational informativeness filter on reduced-form residuals rather than a necessary
identification condition. Layer 3: LiNGAM on the VAR residuals (DirectLiNGAM for Compustat, ICA--LiNGAM for
TEP; \citealp{shimizu2011}; \citealp{hyvarinen2010}) recovers $\hat B,\hat e_t$. For Compustat,
the observed series are standardized to unit variance before VAR estimation; for TEP, the VAR
residuals are standardized to unit variance before ICA. Recovered contemporaneous coefficients
are therefore expressed on standardized scales rather than in heterogeneous original units; for
TEP they attach to unit-variance residuals, so magnitudes are directly comparable across edges.
Independence
support is assessed by MI, block-permutation dHSIC ($\ell=4$ for the annual Compustat series,
$\ell=10$ for TEP; $499$ and $999$ permutations) and tail co-exceedance ($q=0.90$ for Compustat
because $q=0.95$ leaves two or three exceedances on $n\approx45$; $q=0.95$ for TEP). Layer 4:
persistence by moving-block bootstrap ($500$ resamples, block length $10$), recording stability
(share of resamples in which the edge is retained with $|\hat b|\ge\eta$) and sign consistency
(share of those with the sign of the point estimate);\footnote{In the stored Compustat and TEP tables the sign share was
computed as the share of the dominant sign among the retained resamples. Recomputed with the sign of the point
estimate from the saved bootstrap output, no persistence label changes at $(0.70,0.80)$ in either application:
every lag-2 Compustat value is unchanged, and one lag-1 firm and four TEP edges change in value without crossing
the threshold. The Monte Carlo uses the point-estimate sign throughout.} an edge is persistent if stability $\ge0.70$
and sign consistency $\ge0.80$. We use $\eta=0.02$ as a permissive baseline cutoff in the
empirical work;
because every Compustat leading edge has $|\hat b|\ge0.47$ and mean absolute bootstrap
coefficient $\ge0.50$, the cutoff is not binding for the leading-edge classification below:
recomputing stability from the saved bootstrap coefficients, the set of persistent leading edges is
identical for every $\eta\in[0.02,0.20]$ (Section~\ref{sec:survive}). Section~\ref{sec:mc} shows
that for weak edges the cutoff should instead be calibrated to bootstrap noise.

\subsection{Admissibility}
Of $11{,}945$ candidate Compustat firm IDs, $60$ pass the combined admissibility and
distribu\-tional-infor\-mativeness screens of Layers~1 and~2 ($0.50\%$). The
dominant exclusion is insufficient non-zero observations after transformation. The accepted units
are not neutral: their median post-differencing length is $46.5$ observations ($44.5$ after the two
VAR lags) versus $4$ for rejected units, with larger scale and complete R\&D reporting,
concentrated in data-rich industrial SIC groups ($36$, $28$, $38$, $35$, $37$). Compustat results
describe this admissible subset, a documented boundary of the method. Because the exercise is a
collection of firm-level systems, we do not summarize it by an aggregate causal graph, which would
invite a structural reading at a level the protocol has not established.

\section{Monte Carlo evidence under known ground truth}
\label{sec:mc}
Before the data, we ask whether the protocol's two focal post-recovery layers behave as claimed when the
truth is known. The design is the bivariate stationary SVAR of Proposition~\ref{prop:constr},
$x_t=\phi x_{t-1}+e^x_t$, $y_t=b\,x_t+\phi y_{t-1}+e^y_t$, with $\phi=0.3$ and standardized $t_3$
shocks. Independence support is governed by a common stochastic-volatility loading $\lambda$
($\lambda=0$: independent shocks; $\lambda>0$: uncorrelated but dependent shocks), persistence by
the edge strength $b$. Each replication runs the full recovery-and-diagnosis pipeline: OLS VAR(1),
ICA--LiNGAM, block-permutation dHSIC ($499$ permutations, $\ell=\lfloor\sqrt T\rfloor$) with
support declared at $p>0.10$, and moving-block-bootstrap persistence ($200$ resamples, thresholds
$(0.70,0.80)$). The simulation validates the protocol's logic in this bivariate ICA--LiNGAM
implementation; it does not reproduce the four-variable DirectLiNGAM estimator or the admission
screens of the Compustat design, so the empirical dHSIC $p$-values of
Section~\ref{sec:nonnest} are additionally conditional on admission.
Table~\ref{tab:mc-main} summarizes; \OA{C} gives full tables.

\begin{table}[t]
\caption{\label{tab:mc-main}Monte Carlo summary.}
\begin{center}
\small
Panel A: size and power of the support diagnostic (rejection rate at the $10\%$ level)\\[2pt]
\begin{tabular*}{0.8\textwidth}{@{}l@{\extracolsep{\fill}}ccc@{}}
\hline\hline
 & $T=50$ & $T=100$ & $T=250$ \\
\hline
Size ($\lambda=0$)   & 0.069 & 0.068 & 0.105 \\
Power ($\lambda=1$)  & 0.181 & 0.367 & 0.652 \\
Power ($\lambda=2$)  & 0.680 & 0.945 & 1.000 \\
\hline\hline
\end{tabular*}

\medskip
Panel B: persistence tracks edge strength, support does not ($T=250$, independent shocks)\\[2pt]
\begin{tabular*}{0.8\textwidth}{@{}l@{\extracolsep{\fill}}cccccc@{}}
\hline\hline
True edge $b$ & $0$ & $0.02$ & $0.05$ & $0.10$ & $0.20$ & $0.50$ \\
\hline
$\Pr(\widehat P=1)$, $\eta=0.02$ & 0.466 & 0.512 & 0.592 & 0.798 & 0.918 & 0.984 \\
$\Pr(\widehat P=1)$, $\eta=0.20$ & 0.008 & 0.006 & 0.010 & 0.028 & 0.328 & 0.980 \\
$\Pr(\widehat S=1)$              & 0.912 & 0.910 & 0.920 & 0.906 & 0.896 & 0.918 \\
\hline\hline
\end{tabular*}

\medskip
Panel C: protocol assignment, $\eta=0.20$, $T=250$ (rows: designed cell; columns: assigned cell)\\[2pt]
\begin{tabular*}{0.8\textwidth}{@{}l@{\extracolsep{\fill}}cccc@{}}
\hline\hline
Designed cell & S\,\&\,P & S\,\&\,$\bar{\mathrm P}$ & $\bar{\mathrm S}$\,\&\,P & $\bar{\mathrm S}$\,\&\,$\bar{\mathrm P}$ \\
\hline
S\,\&\,P ($\lambda=0$, $b=0.9$)                          & {\bf 445} & 3 & 48 & 4 \\
S\,\&\,$\bar{\mathrm P}$ ($\lambda=0$, $b=0$)             & 3 & {\bf 439} & 0 & 58 \\
$\bar{\mathrm S}$\,\&\,P ($\lambda=2$, $b=0.9$)           & 0 & 0 & {\bf 420} & 80 \\
$\bar{\mathrm S}$\,\&\,$\bar{\mathrm P}$ ($\lambda=2$, $b=0$) & 0 & 0 & 78 & {\bf 422} \\
\hline\hline
\end{tabular*}
\end{center}
\footnotesize
\renewcommand{\baselineskip}{11pt}
\textbf{Note:} Bivariate SVAR(1) with $\phi=0.3$ and standardized $t_3$ shocks; $\lambda$ is the
common stochastic-volatility loading, $b$ the true instantaneous edge. Full recovery-and-diagnosis pipeline in every
replication (OLS VAR(1), ICA--LiNGAM, block-permutation dHSIC with $499$ permutations and
$\ell=\lfloor\sqrt T\rfloor$, moving-block bootstrap with $200$ resamples). Panel A: $1{,}000$
replications per cell, $b=0.9$, Monte Carlo standard error at most $0.016$. Panels B and C: $500$
replications per column (B) or row (C), Monte Carlo standard error at most $0.022$. $\widehat P$:
leading edge declared persistent; $\widehat S$: independence support declared ($p>0.10$). S: support;
P: persistent; a bar denotes absence.
\end{table}

Three lessons follow. First (Panel~A), the block-permutation dHSIC diagnostic on recovered shocks
has empirical size close to or below nominal at every $T$, and its power against shared
stochastic volatility rises from moderate at $T=50$, the Compustat regime, to essentially one at
$T=250$. Non-rejection at $n\approx45$ is therefore weak positive evidence, and we read the
Compustat support classification with that in mind. Second (Panel~B), independence support is
flat in $b$ at about $1-\alpha$ while persistence rises with $b$: the two axes track different
features of the process. Panel~B also shows the cutoff problem: with $\eta=0.02$ a non-existent
edge is declared persistent in $47\%$ of replications, because the leading recovered coefficient
under $b=0$ is pure sampling noise with bootstrap standard deviation about $0.11$, far above the
cutoff. Setting $\eta=0.20$, about two bootstrap standard deviations, reduces this to under $1\%$
while leaving a genuine edge of $b=0.5$ persistent $98\%$ of the time. Third (Panel~C), with the
calibrated cutoff the protocol assigns both off-diagonal cells, the content of non-nestedness,
correctly in $88\%$ and $84\%$ of replications; the support axis is correctly classified in
$94.5\%$ and the persistence axis in $91.6\%$ of all replications. The residual errors concentrate on the persistence axis of the dependent-shock row and run in both
directions, $78$ and $80$ of $500$ replications crossing between the two dependent-shock cells of
Panel~C: the common-volatility factor amplifies tail dispersion and raises bootstrap noise, which
both declares some non-existent edges persistent and destabilizes some true ones. This lowers
four-cell classification accuracy but does not overturn either off-diagonal cell.

The designs above draw shocks that are independent over time, so they do not stress the serial
dependence that motivates the block permutation. A fourth design therefore puts persistence into
the shocks themselves through persistent log-volatilities, independent across shocks under the
null and common under the alternative (\OA{C}). Under this serially dependent null the
block-permutation diagnostic keeps approximately nominal size at the sample lengths relevant for
the application ($0.088$ at $T=50$ and $0.102$ at $T=100$) and drifts to $0.134$ at $T=250$,
while power against the common-volatility alternative rises from $0.158$ to $0.652$; the free
permutation behaves similarly in this design ($0.092$ to $0.138$ under the null), in contrast
with the recovered TEP shocks, where the free-permutation reference is substantially narrower
than the block-permutation reference (\OA{B}). The residual size drift at larger $T$ is consistent with the caveat that the
diagnostic is not an exact test on recovered shocks (Section~\ref{sec:disc}).

\section{Empirical results}
\label{sec:res}
\subsection{TEP: coefficient magnitude and persistence}
In TEP the two largest recovered contemporaneous coefficients, $\texttt{xmeas\_10}\to
\texttt{xmeas\_36}$ ($|\hat b|=3.31$) and $\texttt{xmv\_6}\to\texttt{xmeas\_36}$ ($3.28$), have
bootstrap stability $0.18$ and $0.02$ (Table~\ref{tab:tep}). Three edges pass the persistence rule:
$\texttt{xmeas\_29}\to\texttt{xmeas\_35}$ (stability $0.80$, absolute coefficient $0.20$),
$\texttt{xmeas\_27}\to\texttt{xmeas\_36}$ ($0.77$, $0.20$), and $\texttt{xmeas\_35}\to\texttt{xmeas\_36}$
($0.72$, $0.18$).

The two large coefficients involve the near-collinear innovation pair $\texttt{xmeas\_10}$ and
$\texttt{xmv\_6}$. When either variable is omitted and the VAR is re-estimated, the corresponding edge
between the retained variables is absent from the point fit.
In both $13$-variable ablations, the joint dHSIC diagnostic still rejects independence, and the same three
edges pass the stated persistence rule under the seed-0 pipeline. Point-fit edge sets and some orientations
also vary across the five ICA initialisations examined. Persistence and independence diagnostics were not
recomputed for each initialisation. Figure~\ref{fig:tepgraph} shows the seed-0 point fit of the
$14$-variable specification before bootstrap filtering; the sensitivity check is reported in \OA{F}.

\begin{table}[t]
\caption{\label{tab:tep}TEP: coefficient magnitude versus persistence, selected edges.}
\begin{center}
\begin{tabular*}{0.7\textwidth}{@{}l@{\extracolsep{\fill}}lcc@{}}
\hline\hline
Source & Target & $|\hat b|$ & Stability \\
\hline
\texttt{xmeas\_10} & \texttt{xmeas\_36} & 3.311 & 0.18 \\
\texttt{xmv\_6}    & \texttt{xmeas\_36} & 3.278 & 0.024 \\
\texttt{xmeas\_29} & \texttt{xmeas\_35} & 0.201 & {\bf 0.80} \\
\texttt{xmeas\_27} & \texttt{xmeas\_36} & 0.197 & {\bf 0.77} \\
\hline\hline
\end{tabular*}
\end{center}
\footnotesize
\renewcommand{\baselineskip}{11pt}
\textbf{Note:} Stability is the share of $500$ moving-block-bootstrap resamples (block length
$10$) in which the edge is retained. Bold entries pass the persistence rule (stability $\ge0.70$,
sign consistency $\ge0.80$).
\end{table}

\begin{figure}[t]
\centering
{\includegraphics[width=0.72\textwidth]{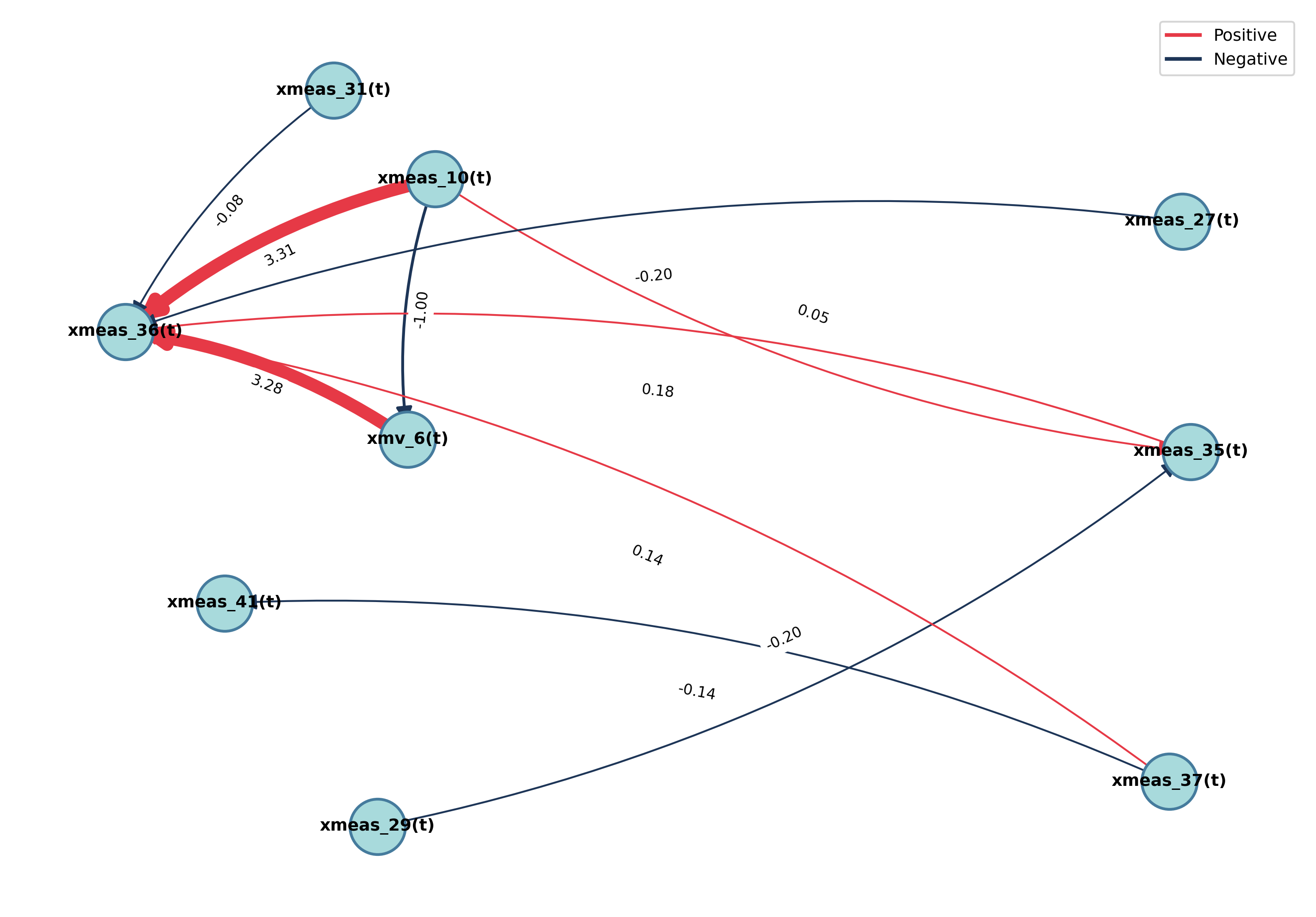}}
\caption{\label{fig:tepgraph}TEP instantaneous graph before bootstrap filtering (seed-0 point fit of the $14$-variable specification).}
\end{figure}

\subsection{Independence support}
\label{sec:indep}
For TEP, recovered-shock MI is low after decomposition (off-diagonal mean $0.037$, against
substantially higher MI among raw residuals; \OA{D}). The joint dHSIC statistic
($0.00490$) exceeds every draw of the block-permutation reference ($\ell=10$, $999$ permutations,
null mean $0.00233$, standard deviation $0.00020$), so the diagnostic rejects joint independence
at $p=0.001$, the smallest value attainable with $999$ permutations, even though average dependence is low, because weak dependence spread across many
pairs accumulates. Pairwise, $58$ of $91$ pairs are individually consistent
with independence at the $10\%$ level and $33$ are not. Maximum upper- and lower-tail
co-exceedance are $0.243$ and $0.284$ against an independence benchmark of $0.05$; with few joint
tail observations we treat these as descriptive. TEP thus has a favourable average dependence
profile, but joint independence is not supported.

For Compustat, MI falls from a mean of $0.340$ on raw residuals to $0.061$ on the recovered
shocks across the $60$ IDs, with wide heterogeneity: many IDs near zero and a right tail of
high-dependence IDs (Figure~\ref{fig:cmi}; the left panel ranks IDs by mean recovered-shock MI
with the median-MI ID marked by the dashed line, the right panel plots raw-residual MI against
recovered-shock MI with most points below the diagonal). The block-permutation dHSIC statistic
ranges over $0.009$--$0.062$; the $p$-value distribution is dispersed (median $0.097$), with
independence not rejected ($p>0.10$) for $29$ of $60$ IDs and rejected for $31$. The tail
diagnostic at $q=0.90$ has a median maximum upper-tail co-exceedance of $0.260$ (benchmark
$0.10$); only $4$ IDs pass a tail screen, each reflecting a near-absence of exceedances rather than
demonstrated tail independence; the overlap between dHSIC non-rejection and the tail screen is
$2$ of $60$, which we report descriptively.
Table~\ref{tab:indep} summarizes.

\begin{figure}[t]
\centering
{\includegraphics[width=0.49\textwidth]{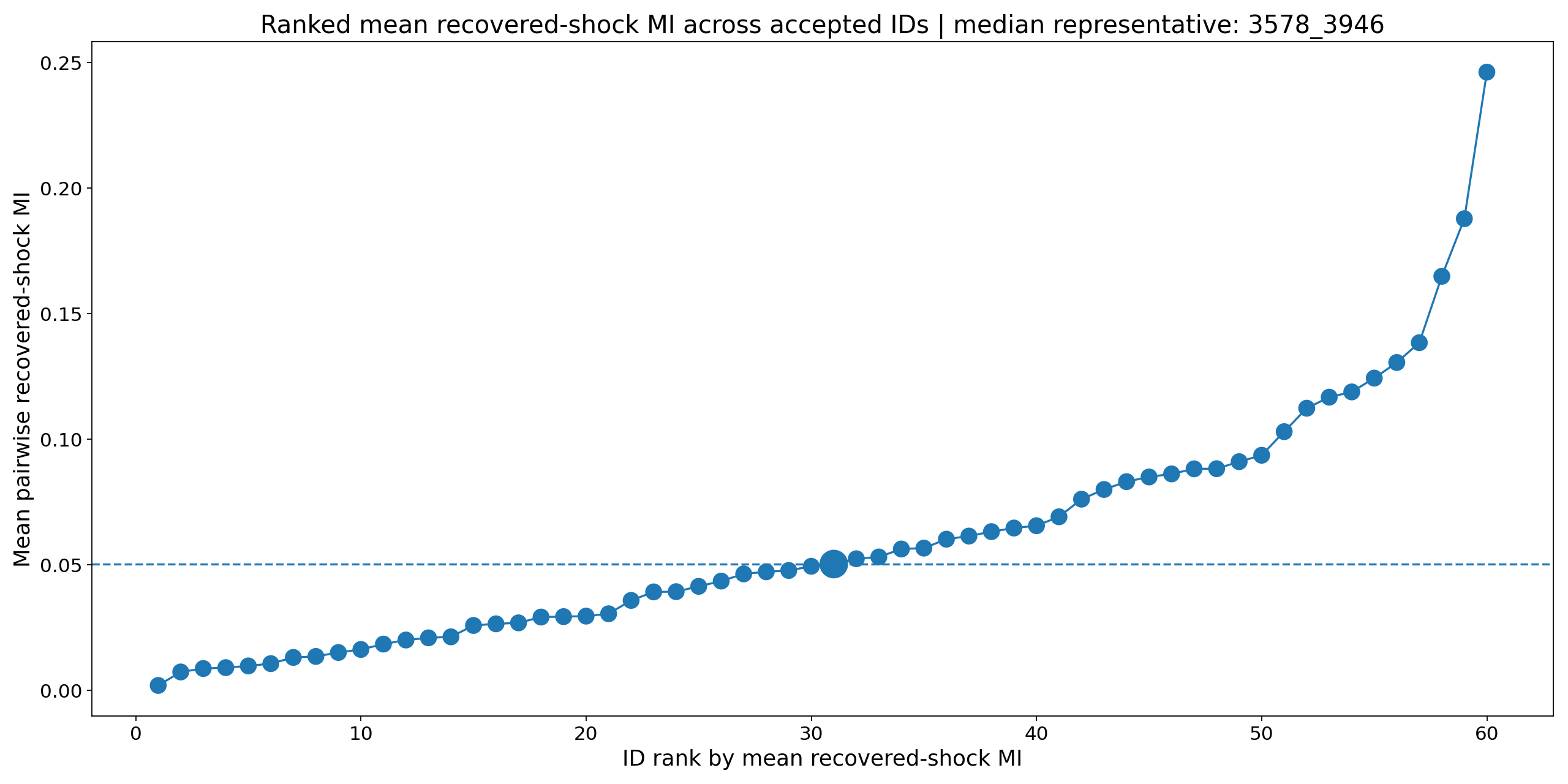}}\hfill
{\includegraphics[width=0.49\textwidth]{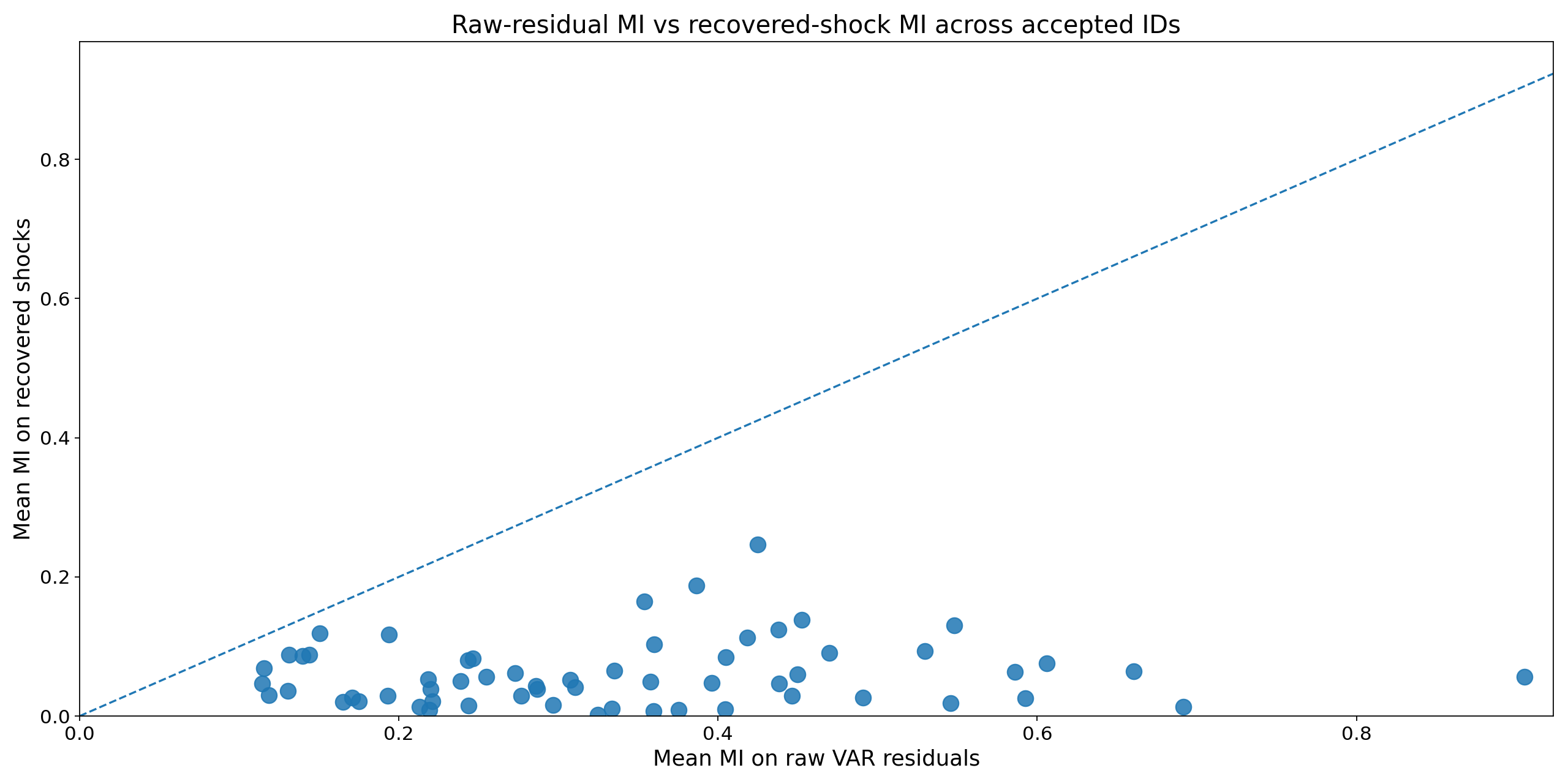}}
\caption{\label{fig:cmi}Compustat recovered-shock MI across the $60$ accepted IDs.}
\end{figure}

\begin{table}[t]
\caption{\label{tab:indep}Independence-support summary.}
\begin{center}
\small
\begin{tabular*}{\textwidth}{@{}l@{\extracolsep{\fill}}cc@{}}
\hline\hline
 & \shortstack[c]{TEP\\($14$ shocks, $T=493$)} & \shortstack[c]{Compustat\\($60$ IDs, median $n=44.5$)} \\
\hline
Recovered-shock MI (mean) & 0.037 & 0.061 (from 0.340 raw) \\
Joint dHSIC statistic & 0.00490 & 0.009--0.062 \\
Block-permutation dHSIC $p$ & 0.001 ($\ell=10$) & median 0.097; 29 with $p>0.10$ ($\ell=4$) \\
Max tail $\widehat\chi$ (upper/lower) & 0.243/0.284 ($q=0.95$) & median 0.260 ($q=0.90$) \\
Independence support (dHSIC) & not supported & supported for 29/60 \\
\hline\hline
\end{tabular*}
\end{center}
\footnotesize
\renewcommand{\baselineskip}{11pt}
\textbf{Note:} TEP is a single $14$-variable system; Compustat entries are distributions across
$60$ accepted IDs. The tail quantile is $q=0.90$ for Compustat because $q=0.95$ leaves two or three
exceedances on $n\approx45$. The descriptive dHSIC-and-tail overlap contains $2$ of the $60$
Compustat IDs.
\end{table}

\subsection{Persistence and independence support are non-nested}
\label{sec:nonnest}
We classify the $60$ Compustat IDs by (independence support) $\times$ (leading-edge persistence),
taking support as dHSIC non-rejection ($p>0.10$) and persistence as the rule of
Definition~\ref{def:pers} applied to each ID's largest-magnitude recovered edge, the edge an
investigator would report (Table~\ref{tab:2x2}). Both off-diagonal cells are populated: eleven IDs
have a persistent leading edge without independence support, and twenty have independence support
without a persistent leading edge. The classification is robust to its ingredients. Replacing
``leading edge persists'' with ``any recovered edge persists'' enlarges the persistent column from
$20$ to $31$ but leaves both off-diagonal cells non-empty ($15$ and $17$). As descriptive
sensitivity checks, cross-tabulating persistence with the MI low-dependence subset gives
off-diagonal counts of $12$ and $17$, and with the dHSIC-and-tail overlap $2$ and $20$; the
dHSIC classifier remains the operational rule, and in every version the supported-and-persistent
cell is small ($9$, $3$ and $0$). Varying the persistence thresholds over stability
$\{0.60,0.70,0.80\}$ and sign consistency $\{0.70,0.80,0.90\}$, recomputed from the stored
bootstrap coefficients, changes counts but not the pattern: both off-diagonal cells remain
populated in all nine combinations, with the persistent-but-unsupported cell ranging over
$4$--$15$ IDs and the supported-but-not-persistent cell over $17$--$28$.

\begin{table}[t]
\caption{\label{tab:2x2}Non-nestedness of persistence and independence support, Compustat.}
\begin{center}
\begin{tabular*}{0.8\textwidth}{@{}l@{\extracolsep{\fill}}ccc@{}}
\hline\hline
 & Not persistent & Persistent & Total \\
\hline
Independence support ($p>0.10$)      & 20 & 9  & 29 \\
No independence support ($p\le0.10$) & 20 & 11 & 31 \\
\hline
Total                                & 40 & 20 & 60 \\
\hline\hline
\end{tabular*}
\end{center}
\footnotesize
\renewcommand{\baselineskip}{11pt}
\textbf{Note:} $60$ accepted IDs. Rows: block-permutation dHSIC support ($\ell=4$, $499$
permutations). Columns: persistence of the leading (largest-magnitude) recovered edge under
Definition~\ref{def:pers} with $\eta=0.02$, stability $\ge0.70$ and sign consistency $\ge0.80$. The
leading-edge classification is identical for every $\eta\in[0.02,0.20]$ and changes by one edge at
the noise-calibrated $\eta=0.30$ (Section~\ref{sec:survive}).
\end{table}

The partition is between two \emph{diagnostics}, each computed with sampling error, and
Section~\ref{sec:mc} shows that at $n\approx45$ the support diagnostic has limited power against
moderate dependence. The empirical cells are therefore read descriptively, not as unbiased
estimates of population cells; the population distinction rests on Proposition~\ref{prop:constr}
and its finite-sample recoverability on Table~\ref{tab:mc-main}. What the data add is that the two
post-recovery quantities considered here, bootstrap persistence and measured independence
support, come apart in the same firm-level systems, so one cannot stand in for the other.

The two off-diagonal cells also carry asymmetric interpretive weight. A persistent edge without
independence support is the dangerous configuration: a highly reproducible relation with no
empirical entitlement to a structural reading. Independence support without a persistent leading
edge is not its mirror image and need not indicate a credibility failure: it means that the
leading recovered relation is not reproducible under the stated rule, and other recovered edges
may nevertheless be persistent (the any-edge robustness check above implies that $5$ of these
$20$ systems retain some persistent edge). In the Monte Carlo the corresponding cell is generated
by a correctly specified zero-edge process (Section~\ref{sec:mc}), for which the absence of a
reportable edge is the right answer. Layer~4 thus filters
what may be reported, while Layer~3 governs what may be interpreted.

\subsection{Candidate structural standing and robustness}
\label{sec:survive}
Under the operational classifier, the supported-and-persistent cell of Table~\ref{tab:2x2}
contains $9$ of $60$ systems. The descriptive measures point the same way. Ranking systems by MI
and taking the low-dependence quartile ($15$ IDs), the persistence rule applied to their $57$
recovered firm-edge observations retains $8$ firm-edges, concentrated in the operating-performance
block (gp--revt), with heterogeneous direction across firms: gp$\to$revt appears in $27$ IDs and
revt$\to$gp in $32$ in the baseline, $6$ versus $9$ in the low-MI quartile and $4$ versus $1$ among
the retained firm-edges. Sign consistency is high \emph{within} each retained firm-edge (mean
$\approx1.0$), so the pattern is consistent with cross-firm heterogeneity rather than within-firm
instability. The $2$ IDs in the descriptive dHSIC-and-tail overlap have no persistent leading
edge. Systems combining no detected independence violation with a reportable persistent leading
edge are therefore rare in this sample. The
economic reading is strictly downstream: gross profit is an accounting function of revenue, but
that identity does not fix the ordering of their innovations.

Four robustness checks bear directly on the protocol's operational choices. Magnitude cutoff:
because the bootstrap stores every retained coefficient, the persistence classification can be
recomputed for any cutoff without re-estimation. The median bootstrap standard deviation of a
Compustat leading coefficient is $0.16$ (interquartile range $0.11$--$0.29$), so the
two-standard-deviation rule of Section~\ref{sec:mc} suggests $\eta\approx0.3$. The set of $20$
persistent leading edges is identical for every $\eta\in\{0.02,0.10,0.20\}$; at $\eta=0.30$ one
edge drops out (Table~\ref{tab:2x2} becomes $20/9/21/10$); even at $\eta=0.50$ both off-diagonal
cells remain populated ($23$ and $8$, with $6$ supported-and-persistent). Non-nestedness in the
data does not depend on the cutoff. Classification level: recomputing the support column of
Table~\ref{tab:2x2} at $\alpha=0.05$ from the stored $p$-values gives $24/11/16/9$ (support and
not persistent, support and persistent, no support and not persistent, no support and
persistent): both off-diagonal cells remain populated at the $5\%$ and $10\%$ levels. Bootstrap
block length: recomputing the moving-block bootstrap at block lengths $l_b\in\{4,6,8\}$ ($500$
resamples per firm, exact re-run of the persistence pipeline) gives $2\times2$ tables of
$23/6/21/10$, $22/7/24/7$ and $21/8/23/8$: both off-diagonal cells remain populated at every
block length, and persistence labels agree with the $l_b=10$ baseline for $52$--$54$ of the $60$
firms. Lag order:
a lag-1 specification admits $73$ IDs ($0.61\%$),
overlapping the baseline in $33$. Its residual layer retains more dependence than the lag-2
baseline (dHSIC statistics $0.011$--$0.058$, median $p=0.014$, support for $19$ of $73$ IDs
against $29$ of $60$), while $32$ of $73$ leading edges are persistent, again for every
$\eta\le0.30$. The lag-1 counterpart of Table~\ref{tab:2x2} is $10/9/31/23$ (support and not
persistent, support and persistent, no support and not persistent, no support and persistent):
both off-diagonal cells are again populated, the table is unchanged for every $\eta\le0.30$, the
persistent-but-unsupported cell is larger than at lag 2, and the operating-performance block
persists across specifications. Because the two lag orders admit different firms, we also compare
them on the $33$ firms admissible under both (\OA{E}). On this common sample the lag-1 and lag-2
tables are $4/3/15/11$ and $8/6/13/6$: both off-diagonal cells are populated at both lag orders,
the leading edge is the same ordered pair for $21$ of $33$ firms (the same unordered pair for
$26$), persistence labels agree for $25$ of $33$ firms, and the lower support count at lag 1 ($7$
against $14$) therefore cannot be attributed to sample composition alone. \OA{E} also gives
the sample flow from $11{,}945$ candidate firms to the $60$ admissible systems.

\section{Discussion: scope and open questions}
\label{sec:disc}
The protocol is diagnostic, not an identification theorem; it presupposes the VAR--LiNGAM result
and asks the post-recovery question. Three items define its frontier. First, formal
inference on the independence of estimated structural shocks already exists in important ICA and
SVAR settings (\citealp{wu2009}; \citealp{amengual2022}; \citealp{amengual2024}). What remains
open here is narrower: the validity of the particular serial-aware diagnostic after the full
VAR--LiNGAM recovery map, including dynamic estimation, structural recovery and data-dependent
ordering, under weak temporal dependence. The block-permutation dHSIC is, by simulation, close to
nominal size on recovered shocks (Table~\ref{tab:mc-main}), but that is finite-sample evidence,
not a theorem. Relatedly, the population pseudo-true coefficient restrictions behind the persistence axis of the
dependent-shock constructions (Assumption~\ref{a:pseudo}) are, for the construction used here, derived exactly
(Appendix~\ref{app:nn}, Proposition~\ref{prop:sym}): for every finite $\lambda$ the axis-aligned orbit is the
unique maximizer of the population contrast, so $\beta^\dagger(0,\lambda)=0$, and for $|b|<2$ the recovered
coefficient is $b$ itself with a strict recovery decision of margin $2-|b|$; the numerical integration of
Table~\ref{tab:a3pop} checks the implementation against these values. The derivation concerns the population
contrast maximized globally, not the sample algorithm. The direction of non-nestedness they
support, persistence without independence support, is the direction that carries the interpretive risk, and
what remains open is the general case: the sample-selected leading edge of a system with more than two
variables requires a margin condition, a unique population leading edge separated from the runner-up, and
the behaviour of the recovery decision under dependent shocks is not derived beyond the two-variable
construction; and the operating characteristics of the persistence rule at a cutoff fixed in bootstrap-noise
units are left to later work. Second, tail co-exceedance is nonparametric and,
in the short Compustat series, informative at the aggregate rather than per-firm level; a
size-controlled tail-independence diagnostic is left open. Third, the logic of recovery,
persistence and independence support as distinct, non-nested objects extends to any
conditional-identification procedure in which recovery is easier than structural entitlement,
including ICA-based and higher-moment SVARs; we restrict attention to VAR--LiNGAM as an established
causal-discovery case. The directional-standing standard of Section~\ref{sec:standing} is stated as a
standard of evidence; a worked application at the scale of a large financial system, where the three
assessments can be run out of sample, is the natural test of whether it is demanding in the right places.

\section{Conclusion}
\label{sec:conc}
Distribution-based recovery returns an object that looks structural, and the bootstrap can make
it look stable, without either step confronting the assumption that licenses a structural reading.
We organize the move from recovery to interpretation as a four-layer protocol, show that its last
two layers, bootstrap persistence and independence support, are non-nested in population, in
simulation and in the Compustat application, and make independence support a jointly reported
quantity alongside persistence. The two failures are not symmetric: persistence without
independence support does not license a structural claim, whereas the absence of a persistent
leading edge does not, by itself, invalidate an independence-compatible decomposition. Two
operational messages follow for applied use of VAR--LiNGAM: calibrate structural
claims to the alignment of the evidentiary layers, treating independence support as the decisive
post-recovery evidentiary layer, rather than recovery or persistence; and scale the bootstrap
magnitude cutoff to sampling noise, because
otherwise persistence can be declared for edges that do not exist.

\section*{Acknowledgements}
The authors have no financial or other interests to disclose in connection with this research
and received no external funding for it. Replication
code for every reported number (the Compustat and TEP pipelines and the Monte Carlo experiment)
is available from the authors; the Compustat extract is licensed and is not redistributed.

\appendix

\section{A constructive non-nestedness result}
\label{app:nn}
Consider the class $\mathcal D$ of bivariate stationary SVAR(1) processes
\begin{equation}
x_t=\phi x_{t-1}+e^x_t,\qquad y_t=b\,x_t+\phi y_{t-1}+e^y_t,\qquad |\phi|<1,
\label{eq:dgp}
\end{equation}
with $(e^x_t,e^y_t)$ strictly stationary, with finite variance, $E[e^x_te^y_t]=0$ and at least one
non-Gaussian component. No moment condition beyond finite variance is imposed on $\mathcal D$
itself: the convergence properties the argument needs are stated directly as
Assumptions~\ref{a:rec}--\ref{a:test}, and the standardized $t_3$ shocks used in the constructions
and in Section~\ref{sec:mc}, which have finite moments of order below three but no finite fourth
moment, belong to $\mathcal D$. For
$D\in\mathcal D$ let $S_0(D)=1$ if the population dHSIC of the recovered shocks is zero (the
recovered shocks are jointly independent) and $S_0(D)=0$ otherwise. Let $\beta^\dagger(D)$
denote the leading instantaneous coefficient of the \emph{limiting recovered decomposition}, that
is, the probability limit of the recovered leading coefficient; under independent shocks
$\beta^\dagger$ coincides with the true coefficient $b$ (Assumption~\ref{a:rec}), under dependent
shocks it is the pseudo-true value of Assumption~\ref{a:pseudo}. Let $P_0(D)=1$ if
$\beta^\dagger(D)\neq0$ and $P_0(D)=0$ if $\beta^\dagger(D)=0$. Defining persistence on the limiting
recovered object rather than on $b$ keeps $P_0$ aligned with what the bootstrap classifier
estimates. The operational classifiers are $\widehat S_T=1$ when the block-permutation dHSIC does not reject at
level $\alpha$, and $\widehat P_T=1$ when the moving-block-bootstrap stability and sign consistency
of the sample leading edge clear $(\underline s,\underline c)$ with magnitude cutoff $\eta_T$.

\begin{assumption}\label{a:spec}
The VAR order and the acyclic contemporaneous form are correctly specified.
\end{assumption}
\begin{assumption}\label{a:rec}
If $(e^x_t,e^y_t)$ are i.i.d.\ with mutually independent components and at most one Gaussian
marginal, the VAR--LiNGAM estimator is consistent for $B$ and the recovered shocks converge in
distribution to the true shocks (\citealp{shimizu2006}; \citealp{hyvarinen2010}). We additionally
assume that, under the same conditions, the moving-block bootstrap is consistent for the recovered
leading coefficient: when $b\neq0$ the bootstrap coefficient converges to $b$ with the sign of $b$, and
when $b=0$ the sample and bootstrap leading coefficients satisfy the rate conditions of
Lemma~\ref{lem:weakedge} below.
\end{assumption}
\begin{assumption}\label{a:pseudo}
For the dependent-shock constructions used below, the recovery map converges to a non-singular
pseudo-true decomposition whose leading coefficient $\beta^\dagger(b,\lambda)$ is continuous in the
parameters of \eqref{eq:dgp}, satisfies $\beta^\dagger(0,\lambda)=0$, and at the selected $b\neq0$
is bounded away from zero with the sign of $b$; the moving-block bootstrap consistently reproduces
this pseudo-true limit and, when $\beta^\dagger=0$, the sample and bootstrap leading coefficients
satisfy the rate conditions of Lemma~\ref{lem:weakedge} below.
\end{assumption}
\begin{assumption}\label{a:test}
The dHSIC kernel is characteristic, so the population dHSIC is zero if and only if the recovered
shocks are jointly independent (\citealp{gretton2008}; \citealp{pfister2018}); applied to
consistently recovered i.i.d.\ independent shocks, the block-permutation dHSIC has asymptotic level
$\alpha$ and is consistent against fixed alternatives.
\end{assumption}
Assumption~\ref{a:pseudo} is needed because, under misspecification induced by dependent shocks,
a pseudo-true recovery could in principle create a non-zero leading coefficient when $b=0$ or
attenuate a true edge; continuity alone does not exclude this, so the two properties the proof
uses are stated explicitly. Its role is confined to the persistence axis of the dependent-shock
constructions: on the support axis it enters only through non-singularity of the limit, because
Lemma~\ref{lem:scale} shows that the dependence induced by a shared scale survives every
non-singular linear recovery. Assumption~\ref{a:test} is the high-level condition whose formal
derivation on recovered shocks under mixing remains open (Section~\ref{sec:disc}).
Table~\ref{tab:mc-main} provides finite-sample evidence consistent with the operational
implications of both assumptions for the designs used here: with $\eta=0.20$ the dependent-shock
$b=0$ design is declared
persistent in $16\%$ of replications, against $1\%$ for its independent-shock counterpart, which
is consistent with $\beta^\dagger(0,\lambda)=0$ together with the larger bootstrap noise that
shared volatility induces.

\begin{lemma}\label{lem:weakedge}
Let the limiting leading coefficient be $\beta^\dagger=0$ and let $\hat b_T$, $\hat b^*_T$ denote the sample
and moving-block-bootstrap coefficients of the sample-selected leading edge, with
$\hat b_T=O_p(T^{-1/2})$ and $\hat b^*_T-\hat b_T=O_{p^*}(T^{-1/2})$ in probability. If
$\eta_T\downarrow0$ with $\sqrt T\eta_T\to\infty$, $P^*(|\hat b^*_T|\ge\eta_T)\to_p0$, so the
edge is declared not persistent with probability tending to one. If instead $\eta_T$ is of the
same order as the sampling noise, $\sqrt T\,\eta_T\to c<\infty$, the bootstrap exceedance
probability $P^*(|\hat b^*_T|\ge\eta_T)$ need not vanish, so the edge may be declared persistent
with non-vanishing probability.
\end{lemma}
\begin{proof}
$|\hat b^*_T|\le|\hat b_T|+|\hat b^*_T-\hat b_T|=O_p(T^{-1/2})+O_{p^*}(T^{-1/2})=o_p(\eta_T)$ since
$T^{-1/2}=o(\eta_T)$, giving the first claim; for the second, $\sqrt T|\hat b^*_T|$ and
$\sqrt T\eta_T$ are then of the same stochastic order, so the exceedance probability need not
converge to zero. The finite set of candidate edges is fixed, so the maximum over it preserves the
limit.
\end{proof}
Lemma~\ref{lem:weakedge} separates the cutoff regimes by rate; its finite-sample counterpart is
Panel~B of Table~\ref{tab:mc-main}. At a fixed sample size the same issue arises whenever a fixed
numerical cutoff is small relative to the realized bootstrap noise scale: with $\eta=0.02$ at
$T=250$ the cutoff lies below that scale and the non-existent edge is declared persistent in
$47\%$ of replications, while with $\eta=0.20$, about two bootstrap standard deviations, it is
not.

The second lemma settles the support axis of the dependent-shock constructions: dependence
generated by a shared scale factor cannot be removed by any non-singular linear recovery, so it is
inherited by the recovered shocks whatever the pseudo-true unmixing happens to be.

\begin{lemma}[Shared scale survives linear recovery]\label{lem:scale}
Let $e_t=V_t\,u_t$ with $u_t=(u_{1t},\ldots,u_{dt})'$, $d\ge2$, where the $u_{jt}$ are independent, centred,
non-degenerate Gaussian scale mixtures with finite variance, and $V_t>0$ is non-degenerate with
$E[V_t^2]<\infty$ and independent of $u_t$. Then for every non-singular $d\times d$ matrix $M$ and every
pair of indices $i\neq k$ the components of $\tilde e_t=Me_t$ satisfy, with $w_t=Mu_t$,
\begin{equation*}
\operatorname{Cov}\bigl(|\tilde e_{it}|,|\tilde e_{kt}|\bigr)\;\ge\;
\operatorname{Var}(V_t)\,E|w_{it}|\,E|w_{kt}|\;>\;0,
\end{equation*}
so every pair of recovered shocks is dependent and, the dHSIC kernel being characteristic
(Assumption~\ref{a:test}), the population dHSIC of $\tilde e_t$ is strictly positive. The case $d=2$ is the
one used in the constructions of this appendix; the general case is stated because the applications that
rely on the lemma are large systems. The standardized $t_3$ shocks of the constructions below are such mixtures, and
$V_t=\exp(\tfrac12\lambda z_t)$ with $z_t$ standard normal and $\lambda>0$ satisfies the
conditions on $V_t$.
\end{lemma}
\begin{proof}
Drop the subscript $t$, fix $i\neq k$, and write $u_j=\sigma_jg_j$ with $g_1,\ldots,g_d$ standard normal,
$\sigma_j>0$, and $(\sigma_1,\ldots,\sigma_d,g_1,\ldots,g_d)$ mutually independent; for a Student $t_3$
variable $\sigma_j^2$ is inverse gamma. Since $V>0$ is independent of $w=Mu$, $|\tilde e_i|=|w_i|V$ and
\begin{equation*}
\operatorname{Cov}\bigl(|\tilde e_i|,|\tilde e_k|\bigr)
=E[V^2]\,E|w_iw_k|-(EV)^2\,E|w_i|\,E|w_k|,
\end{equation*}
with all moments finite because $E[w_i^2]<\infty$ and $E[V^2]<\infty$. Conditional on
$\sigma=(\sigma_1,\ldots,\sigma_d)$, $(w_i,w_k)$ is centred bivariate Gaussian with standard deviations
$s_i(\sigma)=(\sum_jm_{ij}^2\sigma_j^2)^{1/2}$ and some correlation $\rho$, so by the absolute-moment
formula of \citet{nabeya1951},
$E[|w_iw_k|\mid\sigma]=\tfrac2\pi s_i(\sigma)s_k(\sigma)\{\rho\arcsin\rho+(1-\rho^2)^{1/2}\}\ge
\tfrac2\pi s_i(\sigma)s_k(\sigma)=E[|w_i|\mid\sigma]\,E[|w_k|\mid\sigma]$, because
$\rho\arcsin\rho+(1-\rho^2)^{1/2}$ has derivative $\arcsin\rho$ and is therefore minimized at
$\rho=0$, where it equals one. Each $E[|w_i|\mid\sigma]=(2/\pi)^{1/2}s_i(\sigma)$ is nondecreasing
in every $\sigma_j$, so the association inequality of \citet{harris1960} for the independent
vector $(\sigma_1,\ldots,\sigma_d)$ gives
$E\bigl[E(|w_i|\mid\sigma)\,E(|w_k|\mid\sigma)\bigr]\ge E|w_i|\,E|w_k|$. Combining the two
inequalities, $E|w_iw_k|\ge E|w_i|\,E|w_k|$, and therefore
\begin{equation*}
\operatorname{Cov}\bigl(|\tilde e_i|,|\tilde e_k|\bigr)\ge\{E[V^2]-(EV)^2\}\,E|w_i|\,E|w_k|
=\operatorname{Var}(V)\,E|w_i|\,E|w_k|,
\end{equation*}
which is strictly positive because $V$ is non-degenerate and each $w_i$ is non-degenerate ($M$ has
no zero row). Dependence of $(|\tilde e_i|,|\tilde e_k|)$ implies dependence of
$(\tilde e_i,\tilde e_k)$ and hence of the vector $\tilde e$, and strict positivity of the population dHSIC
follows from the characteristic-kernel property in Assumption~\ref{a:test}.
\end{proof}
Lemma~\ref{lem:scale} replaces, and proves, the informal statement that no linear unmixing removes
shared-scale dependence; note that it uses only first and second moments of the shocks, so it
applies to the $t_3$ constructions, whose fourth moments are infinite.

\begin{proposition}\label{prop:constr}
(i) Under Assumptions~\ref{a:spec}, \ref{a:rec} and \ref{a:test} there exist processes $D_{SP}$
and $D_{S\bar P}$ in $\mathcal D$ with $(S_0,P_0)=(1,1)$ and $(1,0)$, so the independence-support
set is not contained in the persistent set. (ii) If in addition Assumption~\ref{a:pseudo} holds,
there exist processes $D_{\bar SP}$ and $D_{\bar S\bar P}$ with $(S_0,P_0)=(0,1)$ and $(0,0)$, so
the persistent set is not contained in the independence-support set and the two sets are
non-nested as population objects. In both parts the operational classifiers attain the following
limits: as $T\to\infty$ with $\eta_T\downarrow0$ and $\sqrt T\,\eta_T\to\infty$,
$\Pr(\widehat S_T=1\mid S_0=1)\to1-\alpha$ and $\Pr(\widehat S_T=0\mid S_0=0)\to1$, while
$\Pr(\widehat P_T=1\mid P_0=1)\to1$ and $\Pr(\widehat P_T=0\mid P_0=0)\to1$.
\end{proposition}
\begin{proof}
Two features of \eqref{eq:dgp} are chosen separately: shock dependence governs $S_0$ and the
limiting leading coefficient governs $P_0$. Let $u^x_t,u^y_t$ be i.i.d.\ standardized $t_3$.

Part (i): independent shocks. If $(e^x_t,e^y_t)=(u^x_t,u^y_t)$, the shocks are independent and
non-Gaussian; by Assumption~\ref{a:rec} the recovered shocks converge to them, the population
dHSIC is zero and by Assumption~\ref{a:test} the rejection probability tends to $\alpha$: $S_0=1$
and $\Pr(\widehat S_T=1)\to1-\alpha$. If $b\neq0$, then $\beta^\dagger=b\neq0$ and, by the
bootstrap clause of Assumption~\ref{a:rec}, the recovered leading coefficient and its bootstrap
analogue converge to that non-zero limit with fixed sign, so stability and sign consistency exceed
any
$(\underline s,\underline c)<(1,1)$ with probability tending to one: $P_0=1$, giving $D_{SP}$. If
$b=0$, then $\beta^\dagger=0$; the shocks remain non-Gaussian, so the process is distributionally
informative and Layer~2 passes, but the limiting decomposition has no instantaneous edge and, by
Lemma~\ref{lem:weakedge}, whose rate conditions hold by Assumption~\ref{a:rec}, the sample-selected
leading edge is declared not persistent with probability tending to one: $P_0=0$, giving
$D_{S\bar P}$.

Part (ii): dependent shocks. If instead $e^\bullet_t=u^\bullet_t\exp(\tfrac12\lambda z_t)$ with
$z_t$ i.i.d.\ standard normal independent of $(u^x,u^y)$ and $\lambda>0$, then $e^x_t$ and $e^y_t$
are uncorrelated but dependent through the shared scale. By Lemma~\ref{lem:scale} every
non-singular linear recovery leaves shocks with strictly positive population dHSIC; the recovery
map converges to a non-singular decomposition by Assumption~\ref{a:pseudo}, so $S_0=0$, and by the
consistency in Assumption~\ref{a:test} the rejection probability tends to one:
$\Pr(\widehat S_T=0)\to1$. If $b\neq0$ (at the selected value), $\beta^\dagger(b,\lambda)$ is
bounded away from zero with the sign of $b$ and the bootstrap reproduces this limit
(Assumption~\ref{a:pseudo}), so as in part (i) $P_0=1$, giving $D_{\bar SP}$. If $b=0$, then
$\beta^\dagger(0,\lambda)=0$ (Assumption~\ref{a:pseudo}); Layer~2 again passes, and by
Lemma~\ref{lem:weakedge}, whose rate conditions at $\beta^\dagger=0$ are part of
Assumption~\ref{a:pseudo}, the sample-selected leading edge is declared not persistent with
probability tending to one: $P_0=0$, giving $D_{\bar S\bar P}$.

The four parameter configurations $(\lambda=0,b\neq0)$, $(\lambda=0,b=0)$, $(\lambda>0,b\neq0)$
and $(\lambda>0,b=0)$ therefore populate the four cells. No claim is made that
shock dependence and edge strength act on the two classifiers independently in general; the
finite-sample evidence in Section~\ref{sec:mc} shows that shared volatility does raise bootstrap
noise and hence measured persistence. Existence of the two off-diagonal cells is all
non-nestedness requires; the classifier limits follow from the assumptions as stated. The result is an existence statement
about the population, not a distributional theorem about any finite sample.
\end{proof}
The division of labour between the assumptions is deliberate. Part (i), and with it the direction
that independence support does not imply persistence, uses only Assumptions~\ref{a:spec},
\ref{a:rec} and \ref{a:test}: correct specification, consistency and bootstrap regularity under
independent shocks, and the properties of the test.
Assumption~\ref{a:pseudo} is confined to part (ii), and within part (ii) Lemma~\ref{lem:scale}
settles the support axis for every non-singular limit, so what Assumption~\ref{a:pseudo}
contributes is exactly the persistence classification of the dependent-shock row: a true edge
survives pseudo-true recovery with its sign, no spurious edge is created at $b=0$, and the
bootstrap tracks this limit. Those pseudo-true properties are stated as an assumption. For the shared-scale
construction, Proposition~\ref{prop:sym} derives the exact population pseudo-true coefficient restrictions, and
Lemma~\ref{lem:consist} establishes consistency for an approximate global maximizer; the bootstrap regularity
conditions in Assumption~\ref{a:pseudo}, and convergence of the local deflationary FastICA algorithm, remain
outside these results. Table~\ref{tab:a3pop} checks the numerical implementation against the exact values, and
Table~\ref{tab:mc-main} documents the finite-sample counterparts. The Monte Carlo of
Section~\ref{sec:mc} instantiates exactly
these four processes ($\phi=0.3$, $t_3$ shocks, $b\in\{0,0.9\}$, $\lambda\in\{0,2\}$) and, for the
weak-edge grid, intermediate values of $b$.
\subsection*{Assumption A.3 on the shared-scale construction}
Assumption~\ref{a:pseudo} states, rather than derives, the properties of the pseudo-true
decomposition when the recovery pipeline is applied to dependent shocks. For the bivariate shared-scale
construction used in Proposition~\ref{prop:constr}(ii) and in Section~\ref{sec:mc}, namely $x=e_1$,
$y=b\,x+e_2$ and $e_i=h\,g_i$ with $h=\exp(\tfrac{\lambda}{2}z-\tfrac{\lambda^2}{4})$, $z\sim N(0,1)$
shared and $g_i$ i.i.d.\ standardized $t_3$ (so that $E[h^2]=1$, the shocks are uncorrelated and
$\lambda>0$ makes them dependent), these properties can be derived exactly. The covariance $\Sigma(b)$ is
analytic, so population whitening is exact, and, up to signs and component permutations, any orthogonal
unmixing in two dimensions can be represented by a rotation $R(\theta)$, $\theta\in[0,\pi/2)$. The population pipeline is the map
$\theta^\dagger=\arg\max_\theta\mathcal J(\theta)$ with
$\mathcal J(\theta)=\sum_{i=1}^{2}\bigl(E[G(w_i(\theta)'\tilde z)]-E[G(\nu)]\bigr)^2$, $G=\log\cosh$,
$\nu\sim N(0,1)$ and $\tilde z=\Sigma^{-1/2}(x,y)'$, followed by the diagonal-dominant row permutation,
row normalization, $B^\dagger=I-\widetilde W$ and $\beta^\dagger$ equal to the off-diagonal entry of
$B^\dagger$ of larger magnitude, the decision sequence of the sample pipeline.
\begin{proposition}\label{prop:sym}
(Exact pseudo-true recovery on the shared-scale construction.) For every finite $\lambda\ge0$: (i) at $b=0$,
$\mathcal J(\theta)<\mathcal J(0)$ for every $\theta\notin(\pi/2)\mathbb Z$, so the axis-aligned orbit is the
unique maximizing orbit and $\beta^\dagger(0,\lambda)=0$; (ii) for $|b|<2$ the population maximizer recovers the
axis-aligned source coordinates, $\widetilde W=A_b^{-1}$ with $A_b=\bigl(\begin{smallmatrix}1&0\\ b&1\end{smallmatrix}\bigr)$,
the diagonal-dominance decision is strict with margin exactly $2-|b|$, and $\beta^\dagger(b,\lambda)=b$. With
Gaussian $g_i$ in place of $t_3$ the strict inequality in (i) fails and $\mathcal J$ is constant in $\theta$.
\end{proposition}
\begin{proof}
Write $g_i=\sqrt{S_i}\,\nu_i$ with $\nu_1,\nu_2$ i.i.d.\ $N(0,1)$ and $S_i=1/U_i$, $U_i$ i.i.d.\ $\chi^2_3$, all
independent of each other and of $z$; then $E[S_i]=1$, $S_i$ is non-degenerate, and $e_i=h\sqrt{S_i}\,\nu_i$
with $E[h^2]=1$. (i) At $b=0$, $\Sigma=I$ and $\tilde z=e$. Let $q(t)=\log\cosh\sqrt t$ for $t\ge0$. For $t>0$,
$q''(t)=\bigl(\sqrt t\,\mathrm{sech}^2\sqrt t-\tanh\sqrt t\bigr)/(4t^{3/2})<0$, because $x\mapsto\tanh
x-x\,\mathrm{sech}^2x$ vanishes at zero and has derivative $2x\,\mathrm{sech}^2x\tanh x>0$ for $x>0$; hence
$F(v)=E[G(\sqrt v\,\nu)]=E[q(v\nu^2)]$, $\nu\sim N(0,1)$, is strictly concave on $[0,\infty)$ (all expectations
are finite because $0\le G(x)\le|x|$). Put $a=\cos^2\theta$. Conditionally on $(h,S_1,S_2)$, $w_1(\theta)'e$ is
centred normal with variance $h^2\{aS_1+(1-a)S_2\}$, so
$E[G(w_1(\theta)'e)]=E\bigl[F\bigl(h^2\{aS_1+(1-a)S_2\}\bigr)\bigr]=:A(a)$, and by exchangeability
$E[G(w_2(\theta)'e)]=A(1-a)=A(a)$. For $0<a<1$, strict concavity and $P(S_1\ne S_2)>0$ give
$A(a)>aE[F(h^2S_1)]+(1-a)E[F(h^2S_2)]=A(1)$, while Jensen's inequality and $E[h^2\{aS_1+(1-a)S_2\}]=1$ give
$A(a)\le F(1)=E[G(\nu)]=:C$. Hence $0\le C-A(a)<C-A(1)$ and
$\mathcal J(\theta)=2\{C-A(a)\}^2<2\{C-A(1)\}^2=\mathcal J(0)$ whenever $a\notin\{0,1\}$, that is, whenever
$\theta\notin(\pi/2)\mathbb Z$. At an axis-aligned maximizer the unmixing returns $\pm e_1,\pm e_2$ in some
order, the diagonal-dominance decision compares $2$ with $0$, row normalization gives $\widetilde W=I$ and
$B^\dagger=0$, so $\beta^\dagger(0,\lambda)=0$. If $g_i$ were Gaussian, $S_i\equiv1$, every inequality above is
an equality and $\mathcal J$ is constant. (ii) For general $b$, $(x,y)'=A_be$ with $\Sigma(b)=A_bA_b'$, so
$\tilde z=\Sigma^{-1/2}A_be=O_be$ with $O_b=\Sigma^{-1/2}A_b$ orthogonal, and the contrast of a rotation of
$\tilde z$ is the contrast of a rotation of $e$. By (i), applied to each unit direction, the contrast over
orthonormal pairs is maximized exactly when both directions are axis-aligned in $e$, so every maximizing
unmixing of $(x,y)'$ is $W^\dagger=DPA_b^{-1}$ with
$A_b^{-1}=\bigl(\begin{smallmatrix}1&0\\-b&1\end{smallmatrix}\bigr)$, $P$ a permutation and $D$ a sign matrix.
The absolute-diagonal sum of the correct arrangement is $2$ and that of the alternative is $|b|$, so for
$|b|<2$ the decision is strict with margin $2-|b|$; row normalization then gives $\widetilde W=A_b^{-1}$,
$B^\dagger=I-\widetilde W=\bigl(\begin{smallmatrix}0&0\\ b&0\end{smallmatrix}\bigr)$ and
$\beta^\dagger(b,\lambda)=b$. Nothing in (i)--(ii) depends on $\lambda$ beyond $E[h^2]=1$.
\end{proof}
The proposition is a statement about the population contrast with the $\log\cosh$ non-linearity maximized
globally over rotations, the object the population pipeline computes. It says nothing about the sample
algorithm: the Monte Carlo pipeline (\OA{C}) runs a deflationary FastICA, a local procedure whose initialization
and stopping rule are not part of the proof, and nothing about the bootstrap regularity that
Assumption~\ref{a:pseudo} presumes, which remains a separate assumption.
\begin{lemma}\label{lem:consist}
Let $\hat{\mathcal J}_N$ be the sample contrast on $N$ i.i.d.\ population-whitened draws and
$M_N=N^{-1}\sum_j\|\tilde z_j\|$. Because $|G(u)-G(v)|\le|u-v|$ and $0\le G(x)\le|x|$, the sample $\log\cosh$
means are $M_N$-Lipschitz in $\theta$ and bounded by $M_N$, so $\hat{\mathcal J}_N$ is Lipschitz on $[0,\pi/2]$
with constant at most $4M_N(M_N+C)=O_p(1)$, since $M_N\to_pE\|\tilde z\|<\infty$; a finite net and the
pointwise law of large numbers give $\sup_\theta|\hat{\mathcal J}_N(\theta)-\mathcal J(\theta)|\to_p0$. By
Proposition~\ref{prop:sym} and compactness, for every $\varepsilon>0$,
$\sup\{\mathcal J(\theta):\operatorname{dist}(\theta,\Theta^\dagger)\ge\varepsilon\}<\max_\theta\mathcal J(\theta)$,
where $\Theta^\dagger$ is the maximizing orbit (a positive gap from every point outside the orbit is not
available for a continuous contrast; this separation is the usable condition). Hence any approximate
maximizer, $\hat{\mathcal J}_N(\hat\theta_N)\ge\sup_\theta\hat{\mathcal J}_N(\theta)-o_p(1)$, satisfies
$\operatorname{dist}(\hat\theta_N,\Theta^\dagger)\to_p0$; because the permutation decision on $\Theta^\dagger$ is
strict, with margin $2-|b|$, the sample decision equals the population decision with probability tending to
one and $\hat\beta_N\to_p\beta^\dagger$. Empirical whitening adds a term bounded by
$\|\hat\Sigma^{-1/2}-\Sigma^{-1/2}\|\,N^{-1}\sum_j\|(x_j,y_j)'\|=o_p(1)$. This is the consistency of the recovery
decision that Assumption~\ref{a:pseudo} presumes for the construction, for an estimator that approximately
maximizes the contrast globally; a local algorithm requires that bridge, and the bootstrap consistency in
Assumption~\ref{a:pseudo} is not a consequence of this consistency.
\end{lemma}
Proposition~\ref{prop:sym} also fixes the population angle in the symmetric-whitening convention of the
pipeline, $\theta^\dagger=\arctan(b/2)$ modulo the axis orbit: $8.5308^{\circ}$ at $b=0.3$ and
$24.2277^{\circ}$ at $b=0.9$. Table~\ref{tab:a3pop} is a check of the numerical implementation against these
exact values, not a certificate: $K=8$ independent batches of $N=600{,}000$ draws per $(b,\lambda)$, each with
analytic whitening and a $1^{\circ}$ grid refined to $0.02^{\circ}$ around the maximum, reported as batch
means with standard errors. At $b=0.3$ and $b=0.9$ every entry ($\theta^\dagger$, $\beta^\dagger$ and the
margin) agrees with the exact value within two standard errors. At $b=0$, where the maximizer is identified
only up to the choice between the two axes and each batch locates it to grid precision, $\beta^\dagger$ is
$-0.0007$ (standard error $0.0006$) against $0$ and the margin is $1.997$ ($0.001$) against $2$, the small
shortfall being the finite-grid error of the located angle. The contrast-gap column records that the
maximum exceeds the contrast at every grid point more than $5^{\circ}$ from the maximizing orbit, the
finite-grid shadow of the strict inequality in the proposition. The result sharpens the assumption for
this construction: a shared scale destroys independence and inflates bootstrap noise but does not displace
the coefficient that the filter screens, $\beta^\dagger(b,\lambda)=b$ for every $|b|<2$ and every finite
$\lambda$. That is the population counterpart of the persistent-without-support cell: a real-sized,
persistent edge coexisting with dependent shocks. Not established here are the behaviour of other contrasts
or of the sample FastICA as an algorithm, and the sample-selected leading edge in systems with more than
two variables.

\begin{table}[t]
\caption{\label{tab:a3pop}Population recovery pipeline on the shared-scale construction (batch means, standard errors in parentheses).}
\begin{center}
\small
\begin{tabular*}{\textwidth}{@{}l@{\extracolsep{\fill}}ccccc@{}}
\hline\hline
$b$ & $\lambda$ & $\theta^\dagger$ (degrees) & $\beta^\dagger$ & Permutation margin & Contrast gap \\
\hline
0.9 & 2.0 & 24.25 (0.05) & 0.9003 (0.0009) & 1.097 (0.002) & $4.9\times10^{-4}$ ($4\times10^{-6}$) \\
0.9 & 0.0 & 24.21 (0.02) & 0.8997 (0.0004) & 1.100 (0.001) & $3.6\times10^{-4}$ ($3\times10^{-6}$) \\
0.3 & 2.0 & 8.54 (0.06)  & 0.3001 (0.0010) & 1.698 (0.002) & $5.1\times10^{-4}$ ($1\times10^{-5}$) \\
0.0 & 2.0 & axis-aligned$^{\ast}$ & $-0.0007$ (0.0006) & 1.997 (0.001) & $5.4\times10^{-4}$ ($2\times10^{-5}$) \\
\hline\hline
\end{tabular*}
\end{center}
\footnotesize
\renewcommand{\baselineskip}{11pt}
\textbf{Note:} $K=8$ batches of $N=600{,}000$ draws; $E[G(\nu)]=0.374567$. $^{\ast}$At $b=0$ the maximizer orbit contains both axes, so batch maximizers fall near $0^{\circ}$ or near $90^{\circ}$ and the circular mean of $\theta$ is uninformative, while every batch is axis-aligned and returns $\beta^\dagger\approx0$, as Proposition~\ref{prop:sym} states. The contrast gap is the maximum minus the largest contrast value at grid points more than $5^{\circ}$ (modulo $\pi/2$) from the maximizer.
\end{table}

\newpage
\section*{Supplementary appendices}
\noindent Appendices B--G contain material that is supplementary to the main text: a simulation on
the complementarity of HSIC and tail co-exceedance (B), the free-versus-block permutation
comparison on the TEP shocks (C), the full Monte Carlo tables behind Table~\ref{tab:mc-main},
including the earlier Gaussian-witness design (D), additional TEP figures (E), the Compustat sample flow with the common-sample lag comparison (F), and the sensitivity of the TEP results to the near-collinear innovation pair (G).

\section{Complementarity of HSIC and tail co-exceedance}
\label{app:mc}
Is the tail diagnostic redundant given HSIC? We simulate two bivariate designs with $1{,}000$
replications each. In the null design $X$ and $Y$ are independent standard normal. In the
tail-dependent design $X$ and $Y$ are independent through the body, but about $10\%$ of
observations are forced to be jointly extreme in the upper tail, so joint exceedances occur while
the bulk of the sample is unrelated. For each design we compute the HSIC statistic with a free
permutation reference ($199$ permutations; the data are i.i.d.\ here) and the upper-tail
co-exceedance $\widehat\chi(0.90)$ (Table~\ref{tab:mc}).

\begin{table}[t]
\caption{\label{tab:mc}HSIC rejection rate and mean tail co-exceedance.}
\begin{center}
\begin{tabular*}{0.85\textwidth}{@{}l@{\extracolsep{\fill}}lcc@{}}
\hline\hline
$n$ & Design & HSIC rejection rate ($10\%$) & Mean $\widehat\chi(0.90)$ \\
\hline
$40$ & Independent                 & 0.093 & 0.106 \\
$40$ & Tail-dependent, body-clean  & 0.766 & 0.751 \\
$80$ & Independent                 & 0.092 & 0.093 \\
$80$ & Tail-dependent, body-clean  & 0.958 & 0.788 \\
\hline\hline
\end{tabular*}
\end{center}
\footnotesize
\renewcommand{\baselineskip}{11pt}
\textbf{Note:} $1{,}000$ replications per row. The independence benchmark for
$\widehat\chi(0.90)$ is $0.10$.
\end{table}

At $n=40$, the Compustat regime, HSIC misses the tail dependence in about one quarter of
replications while tail co-exceedance registers values seven times the independence benchmark.
At $n=80$ HSIC regains most of its power. The two diagnostics are complementary: HSIC governs
general dependence and gains power with sample size; tail co-exceedance targets the
extreme-dependence regime that is least visible to a body-of-distribution measure in short series.

\section{Free versus block permutation on the TEP shocks}
\label{app:block}
On the $14$ recovered TEP shocks ($T=493$, block length $\ell=10$, $999$
permutations), the free-permutation null of the joint dHSIC statistic is far tighter than the
block-permutation null: null standard deviations of $0.0000068$ (free) versus $0.00020$ (block),
roughly a thirtyfold difference, with null means of $0.0020$ (free) and $0.0023$ (block). The corresponding
block-to-free standard-deviation ratios are $24.67$ and $18.55$ in the two $13$-variable ablations of
Appendix~\ref{app:tepsens}; the free-permutation reference is narrower in all three specifications examined. This comparison
illustrates the sensitivity of the reference distribution to the resampling scheme on serially
dependent recovered shocks; it does not by itself establish correct null calibration for either
scheme. Size under a known serially dependent null is examined in Appendix~\ref{app:mcfull}, and
block permutation is in any case not an exact test on recovered shocks (\MT{main text, Section~8}{Section~\ref{sec:disc}}).

\clearpage
\suppressfloats[t]
\section{Full Monte Carlo tables}
\label{app:mcfull}
All designs use the bivariate SVAR(1) of \MT{Appendix~A of the main text}{Appendix~\ref{app:nn}} with $\phi=0.3$ and a burn-in
of $100$ observations. Standardized $t_3$ shocks are generated as $z/\sqrt{g}$ with $z$ standard
normal and $g\sim\Gamma(3/2,2/3)$, rescaled to unit variance. Shared stochastic volatility
multiplies both shocks by $\exp(\tfrac12\lambda z_t)$ with $z_t$ standard normal, normalized to unit
mean square. The pipeline in each replication is OLS VAR(1) with intercept; FastICA (tanh
non-linearity, deflation) on the residuals; the permutation of the unmixing matrix that maximizes
the absolute diagonal, row-normalized to give $\hat B=I-\hat W$; block-permutation dHSIC with
Gaussian kernels (median heuristic), block length $\ell=\lfloor\sqrt T\rfloor$, $499$
permutations, support declared at $p>0.10$; moving-block bootstrap with block length
$\lfloor\sqrt T\rfloor$ and $200$ resamples, in each of which the leading edge selected on the
original sample is re-estimated. Persistence for several cutoffs $\eta$ is evaluated from the same
$200$ bootstrap coefficients.

\subsection{Size and power of the support diagnostic}
Table~\ref{tab:sp} reports the rejection rate of the block-permutation dHSIC classifier on recovered shocks at
$T\in\{50,100,250\}$ for independent shocks ($\lambda=0$, size) and for moderate and strong shared volatility
($\lambda\in\{1,2\}$, power).
\begin{table}[t]
\caption{\label{tab:sp}Rejection rate of the block-permutation dHSIC diagnostic on recovered shocks.}
\begin{center}
\begin{tabular*}{0.85\textwidth}{@{}l@{\extracolsep{\fill}}ccc@{}}
\hline\hline
 & $T=50$ & $T=100$ & $T=250$ \\
\hline
$\lambda=0$ (independent shocks; size) & 0.069 & 0.068 & 0.105 \\
$\lambda=1$ (moderate shared volatility) & 0.181 & 0.367 & 0.652 \\
$\lambda=2$ (strong shared volatility)   & 0.680 & 0.945 & 1.000 \\
\hline\hline
\end{tabular*}
\end{center}
\footnotesize
\renewcommand{\baselineskip}{11pt}
\textbf{Note:} Nominal level $10\%$; $1{,}000$ replications per cell; $b=0.9$; Monte Carlo
standard error at most $0.016$.
\end{table}

\subsection{Size and power under serially dependent shocks}
\label{app:serialdep}
The designs above use shocks that are independent over time, so they do not stress the serial
dependence that motivates block permutation. Here the shocks carry persistent log-volatilities:
$e_{it}=\exp(h_{it}/2)\,u_{it}$ with $u_{it}$ i.i.d.\ standardized $t_3$ and
$h_{it}=0.8\,h_{i,t-1}+0.6\,\nu_{it}$, $\nu_{it}$ i.i.d.\ standard normal, so each $h_{it}$ is
stationary with unit variance. Under the null the two volatility processes are independent across
shocks, so the shocks are cross-sectionally independent but serially dependent; under the
alternative a single common volatility process multiplies both shocks. The rest of the design is
unchanged ($\phi=0.3$, $b=0.9$, full recovery-and-diagnosis pipeline in every replication, $499$
permutations, $\ell=\max(3,\lfloor\sqrt T\rfloor)$ for the block permutation and $\ell=1$ for the
free permutation, $\alpha=0.10$, $500$ replications per cell; the replication bundle regenerates
Table~\ref{tab:serialdep}).

\begin{table}[t]
\caption{\label{tab:serialdep}Rejection rates on recovered shocks under serially dependent
shocks: block versus free permutation.}
\begin{center}
\small
\begin{tabular*}{\textwidth}{@{}l@{\extracolsep{\fill}}cccccc@{}}
\hline\hline
 & \multicolumn{3}{c}{\shortstack{Null: independent\\ persistent volatilities}}
 & \multicolumn{3}{c}{\shortstack{Alternative: common\\ persistent volatility}} \\
$T$ & 50 & 100 & 250 & 50 & 100 & 250 \\
\hline
Block permutation & 0.088 & 0.102 & 0.134 & 0.158 & 0.324 & 0.652 \\
Free permutation  & 0.092 & 0.120 & 0.138 & 0.166 & 0.332 & 0.668 \\
\hline\hline
\end{tabular*}
\end{center}
\footnotesize
\renewcommand{\baselineskip}{11pt}
\textbf{Note:} Rejection rates of the dHSIC diagnostic at $\alpha=0.10$ on recovered shocks,
$500$ replications per cell (Monte Carlo standard error at most $0.021$). Null: shocks
cross-sectionally independent with independent persistent log-volatilities. Alternative: common
persistent log-volatility.
\end{table}

Two readings follow. First, at the sample lengths relevant for the Compustat application the
block-permutation diagnostic keeps approximately nominal size under a serially dependent null
($0.088$ at $T=50$, $0.102$ at $T=100$) while retaining power against the common-volatility
alternative ($0.158$ rising to $0.652$ at $T=250$). Second, at $T=250$ both schemes mildly
exceed nominal size ($0.134$ block, $0.138$ free): in this design the serial dependence sits in
second moments only, so the two schemes differ little here, in contrast with the recovered
TEP shocks of Appendix~\ref{app:block}, where the free-permutation reference standard deviation
is about thirty times smaller than that of the block-permutation reference; and the residual
drift under block permutation is consistent with
the caveat of \MT{Section~8 of the main text}{Section~\ref{sec:disc}} that the diagnostic is not
an exact test on shocks that carry dynamic-estimation and structural-recovery error.

\subsection{Weak-edge grid}
Table~\ref{tab:we} varies the true edge strength $b$ at $T=250$ with independent shocks and reports, side by side,
the probability that the leading edge is declared persistent at four cutoffs, the probability that independence is
supported, and the probability of the supported-but-not-persistent cell.
\begin{table}[t]
\caption{\label{tab:we}Weak-edge grid: persistence and support as the true edge strength varies.}
\begin{center}
\small
\begin{tabular*}{\textwidth}{@{}l@{\extracolsep{\fill}}cccccc@{}}
\hline\hline
$b$ & $0$ & $0.02$ & $0.05$ & $0.10$ & $0.20$ & $0.50$ \\
\hline
$\Pr(\widehat P=1)$, $\eta=0.02$ & 0.466 & 0.512 & 0.592 & 0.798 & 0.918 & 0.984 \\
$\Pr(\widehat P=1)$, $\eta=0.05$ & 0.316 & 0.372 & 0.438 & 0.694 & 0.922 & 0.986 \\
$\Pr(\widehat P=1)$, $\eta=0.10$ & 0.096 & 0.092 & 0.138 & 0.340 & 0.846 & 0.984 \\
$\Pr(\widehat P=1)$, $\eta=0.20$ & 0.008 & 0.006 & 0.010 & 0.028 & 0.328 & 0.980 \\
$\Pr(\widehat S=1)$ & 0.912 & 0.910 & 0.920 & 0.906 & 0.896 & 0.918 \\
$\Pr(\widehat S=1,\widehat P=0)$, $\eta=0.02$ & 0.466 & 0.428 & 0.360 & 0.164 & 0.064 & 0.002 \\
$\Pr(\widehat S=1,\widehat P=0)$, $\eta=0.20$ & 0.906 & 0.906 & 0.910 & 0.880 & 0.608 & 0.006 \\
\hline\hline
\end{tabular*}
\end{center}
\footnotesize
\renewcommand{\baselineskip}{11pt}
\textbf{Note:} Independent $t_3$ shocks, $T=250$, $500$ replications per column (Monte Carlo
standard error at most $0.022$). $\widehat P$: leading edge persistent; $\widehat S$: independence
supported. The median bootstrap standard deviation of the leading coefficient is $0.11$ at every
$b$.
\end{table}

\subsection{Confusion matrices for all cutoffs (non-Gaussian witnesses)}
Designed cells: support and persistent ($\lambda=0,b=0.9$); support and not persistent
($\lambda=0,b=0$); no support and persistent ($\lambda=2,b=0.9$); no support and not persistent
($\lambda=2,b=0$). All shocks are $t_3$; $T=250$; $500$ replications per row. Table~\ref{tab:conf}
reports the protocol assignment for each magnitude cutoff $\eta$. Support-axis accuracy is $0.945$
for every $\eta$, because the support classifier does not depend on $\eta$; persistence-axis
accuracy rises from $0.725$ ($\eta=0.02$) to $0.762$, $0.848$ and $0.916$ ($\eta=0.05$, $0.10$,
$0.20$). Median bootstrap standard deviations of the leading coefficient are $0.16$ (support and
persistent), $0.12$ (support and not persistent), $0.42$ (no support and persistent) and $0.26$ (no
support and not persistent): shared stochastic volatility roughly doubles bootstrap noise, which is
why the no-support, not-persistent cell retains $78$ false persistence calls even at $\eta=0.20$.

\begin{table}[t]
\caption{\label{tab:conf}Protocol assignment by magnitude cutoff.}
\begin{center}
\small
\begin{tabular*}{\textwidth}{@{}l@{\extracolsep{\fill}}lcccc@{}}
\hline\hline
$\eta$ & Designed cell & S\,\&\,P & S\,\&\,$\bar{\mathrm P}$ & $\bar{\mathrm S}$\,\&\,P & $\bar{\mathrm S}$\,\&\,$\bar{\mathrm P}$ \\
\hline
$0.02$ & S\,\&\,P & 445 & 3 & 48 & 4 \\
       & S\,\&\,$\bar{\mathrm P}$ & 224 & 218 & 25 & 33 \\
       & $\bar{\mathrm S}$\,\&\,P & 0 & 0 & 419 & 81 \\
       & $\bar{\mathrm S}$\,\&\,$\bar{\mathrm P}$ & 0 & 0 & 214 & 286 \\
\hline
$0.05$ & S\,\&\,P & 445 & 3 & 48 & 4 \\
       & S\,\&\,$\bar{\mathrm P}$ & 153 & 289 & 19 & 39 \\
       & $\bar{\mathrm S}$\,\&\,P & 0 & 0 & 423 & 77 \\
       & $\bar{\mathrm S}$\,\&\,$\bar{\mathrm P}$ & 0 & 0 & 220 & 280 \\
\hline
$0.10$ & S\,\&\,P & 445 & 3 & 48 & 4 \\
       & S\,\&\,$\bar{\mathrm P}$ & 45 & 397 & 4 & 54 \\
       & $\bar{\mathrm S}$\,\&\,P & 0 & 0 & 422 & 78 \\
       & $\bar{\mathrm S}$\,\&\,$\bar{\mathrm P}$ & 0 & 0 & 170 & 330 \\
\hline
$0.20$ & S\,\&\,P & 445 & 3 & 48 & 4 \\
       & S\,\&\,$\bar{\mathrm P}$ & 3 & 439 & 0 & 58 \\
       & $\bar{\mathrm S}$\,\&\,P & 0 & 0 & 420 & 80 \\
       & $\bar{\mathrm S}$\,\&\,$\bar{\mathrm P}$ & 0 & 0 & 78 & 422 \\
\hline\hline
\end{tabular*}
\end{center}
\footnotesize
\renewcommand{\baselineskip}{11pt}
\textbf{Note:} Rows: designed cell; columns: cell assigned by the protocol. S: independence
support; P: persistent leading edge; a bar denotes absence. $T=250$; $500$ replications per row.
\end{table}

\subsection{Earlier design with Gaussian not-persistent witnesses}
An earlier version of the experiment instantiated the two not-persistent cells with Gaussian
shocks and $b=0$ ($\lambda=1.5$ for the dependent case) and used $\eta=0.02$. That design is
knife-edge: Gaussianity closes the distribution-based route, so ``not persistent'' there arises
from the unidentified ICA rotation flipping the sign across resamples rather than from the absence
of an edge. Table~\ref{tab:gauss} reports it for transparency; the main text uses the non-Gaussian
witnesses above.

\begin{table}[t]
\caption{\label{tab:gauss}Gaussian-witness design.}
\begin{center}
\small
\begin{tabular*}{\textwidth}{@{}l@{\extracolsep{\fill}}cccc@{}}
\hline\hline
Designed cell & S\,\&\,P & S\,\&\,$\bar{\mathrm P}$ & $\bar{\mathrm S}$\,\&\,P & $\bar{\mathrm S}$\,\&\,$\bar{\mathrm P}$ \\
\hline
S\,\&\,P ($t_3$, $\lambda=0$, $b=0.9$)                          & 448 & 1   & 48  & 3   \\
S\,\&\,$\bar{\mathrm P}$ (Gaussian, $\lambda=0$, $b=0$)          & 33  & 455 & 0   & 12  \\
$\bar{\mathrm S}$\,\&\,P ($t_3$, $\lambda=2$, $b=0.9$)            & 0   & 0   & 418 & 82  \\
$\bar{\mathrm S}$\,\&\,$\bar{\mathrm P}$ (Gaussian, $\lambda=1.5$, $b=0$) & 2   & 1   & 172 & 325 \\
\hline\hline
\end{tabular*}
\end{center}
\footnotesize
\renewcommand{\baselineskip}{11pt}
\textbf{Note:} $\eta=0.02$, $T=250$, $500$ replications per row. Support axis correctly
classified in $96.7\%$ and persistence axis in $85.4\%$ of replications.
\end{table}

\clearpage
\section{Additional TEP figures}
\label{app:tepfig}
Figure~\ref{fig:tepmi} shows the TEP mutual-information matrix among raw VAR residuals (left) and
recovered shocks (right) on a shared colour scale: the decomposition reduces remaining
dependence but does not eliminate it. Figure~\ref{fig:fullgraph} shows the full recovered TEP
VAR--LiNGAM graph before diagnostic filtering, combining the lagged VAR layer (solid) and the
instantaneous LiNGAM decomposition (dashed) for the $14$-variable system (seed-0 point fit). It records the
algorithmic recovery output, not structural evidence; its density is why the main text calibrates
structural interpretation by independence support and bootstrap persistence rather than by
graphical richness.

\begin{figure}[!htbp]
\centering
{\includegraphics[width=0.95\textwidth]{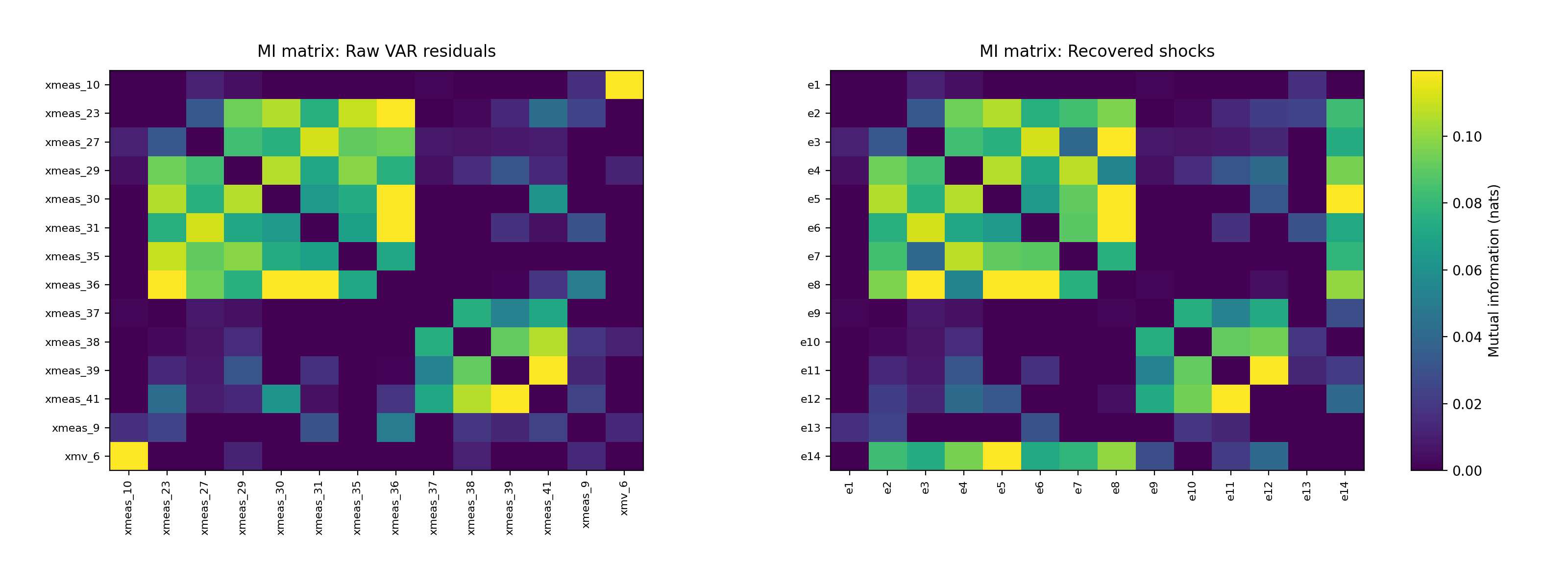}}
\caption{\label{fig:tepmi}TEP mutual information: raw VAR residuals and recovered shocks.}
\end{figure}

\begin{figure}[!htbp]
\centering
{\includegraphics[width=0.98\textwidth]{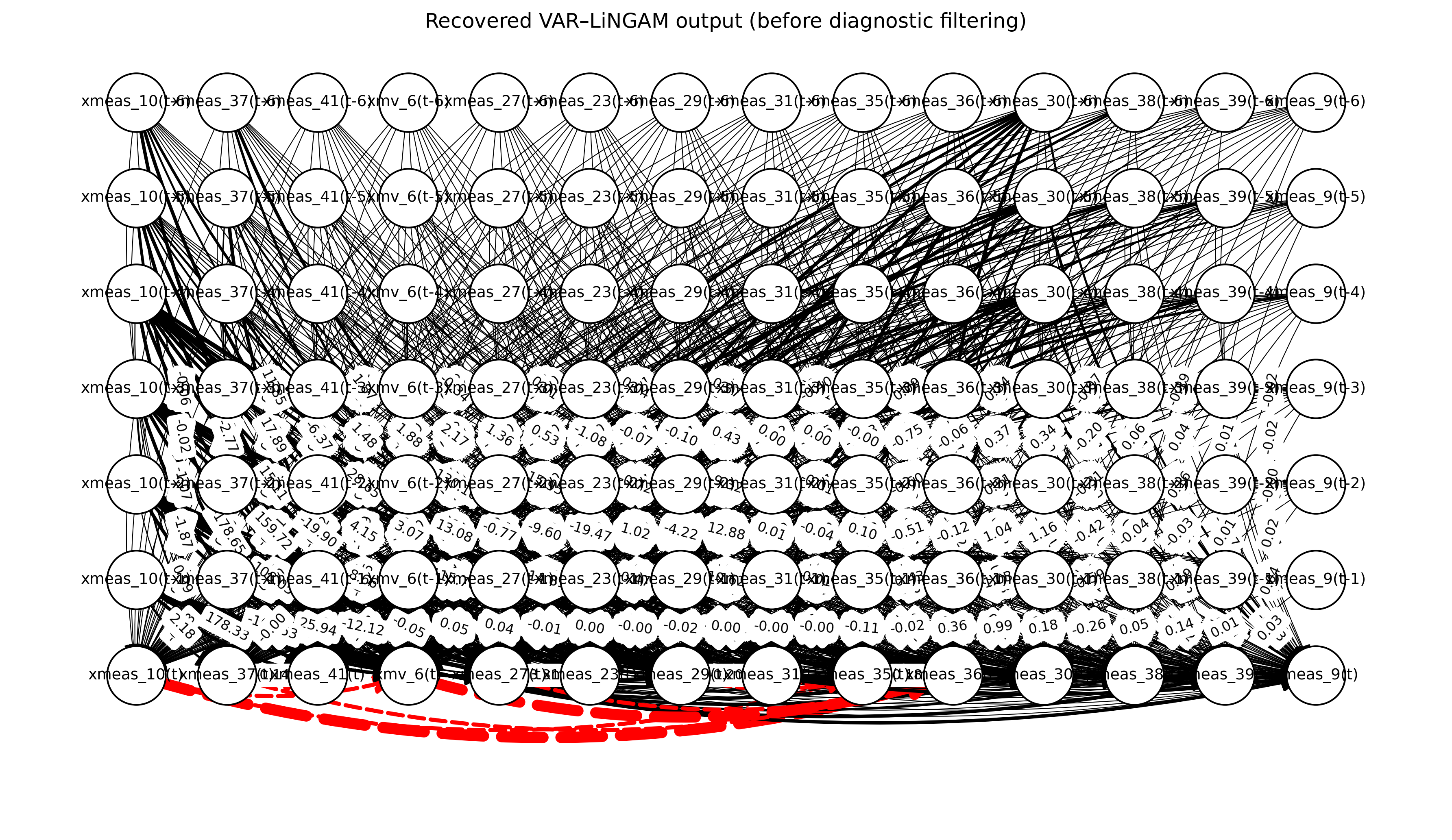}}
\caption{\label{fig:fullgraph}Full recovered TEP VAR--LiNGAM graph before diagnostic filtering.}
\end{figure}

\clearpage
\section{Compustat sample flow and common-sample lag comparison}
\label{app:sample}
Table~\ref{tab:flow} traces the admissibility screen of \MT{Section~5 of the main text}{Section~\ref{sec:data}} from the $11{,}945$ candidate
firm identifiers to the $60$ systems used in the baseline (VAR lag 2). The first two stages are
data-availability conditions; the remaining stages are the reduced-form and distributional
conditions of Layers~1 and~2, applied to every variable of a firm's four-variable system.

\begin{table}[!htbp]
\caption{\label{tab:flow}Compustat sample flow, baseline specification (VAR lag 2).}
\begin{center}
\small
\begin{tabular*}{\textwidth}{@{}p{0.62\textwidth}@{\extracolsep{\fill}}rr@{}}
\hline\hline
Stage & Firms excluded & Firms remaining \\
\hline
Candidate firm identifiers (SIC 2000--3999, 1970--2024) & & 11,945 \\
Fewer than 20 non-zero observations in a variable & 10,639 & 1,306 \\
Fewer than 31 observations after transformation & 749 & 557 \\
Residual autocorrelation (Ljung--Box) in any variable & 12 & 545 \\
Non-stationary residuals (ADF) in any variable & 244 & 301 \\
Non-stationary residuals (KPSS) in any variable & 48 & 253 \\
Gaussian residuals (D'Agostino--Pearson and Jarque--Bera) in any variable & 193 & 60 \\
\hline\hline
\end{tabular*}
\end{center}
\footnotesize
\renewcommand{\baselineskip}{11pt}
\textbf{Note:} Stages are applied sequentially; a firm is excluded at the first stage it fails. The
last four stages are evaluated on the $557$ firms for which a VAR(2) is fitted. The lag-1
specification, with one fewer observation lost to lags, retains $73$ firms.
\end{table}

Because the lag-1 and lag-2 specifications admit different firms ($73$ and $60$, with $33$ in
common), a comparison of their $2\times2$ tables confounds lag order with sample composition.
Table~\ref{tab:common33} therefore repeats the classification of \MT{Section~7.3 of the main text}{Section~\ref{sec:nonnest}} on the $33$ common
firms only.

\begin{table}[!htbp]
\caption{\label{tab:common33}Persistence and independence support on the $33$ firms admissible under both lag orders.}
\begin{center}
\begin{tabular*}{\textwidth}{@{}l@{\extracolsep{\fill}}cc@{}}
\hline\hline
 & VAR lag 1 & VAR lag 2 \\
\hline
Independence support (dHSIC $p>0.10$)                       & 7  & 14 \\
Persistent leading edge                                     & 14 & 12 \\
Support, not persistent                                     & 4  & 8  \\
Support, persistent                                         & 3  & 6  \\
No support, not persistent                                  & 15 & 13 \\
No support, persistent                                      & 11 & 6  \\
Median dHSIC $p$-value                                      & 0.010 & 0.072 \\
\hline
Same leading edge (ordered pair) across lag orders          & \multicolumn{2}{c}{21 of 33} \\
Same leading pair (either orientation)                      & \multicolumn{2}{c}{26 of 33} \\
Same support label across lag orders                        & \multicolumn{2}{c}{20 of 33} \\
Same persistence label across lag orders                    & \multicolumn{2}{c}{25 of 33} \\
Same cell of the $2\times2$ table                           & \multicolumn{2}{c}{16 of 33} \\
\hline\hline
\end{tabular*}
\end{center}
\footnotesize
\renewcommand{\baselineskip}{11pt}
\textbf{Note:} Same firms, same classifiers as \MT{Table~5 of the main text}{Table~\ref{tab:2x2}} ($\ell=4$, $499$
permutations; $500$ bootstrap resamples, $\eta=0.02$, stability $\ge0.70$, sign consistency
$\ge0.80$). Both off-diagonal cells are populated at both lag orders. Persistence labels are
largely stable across lag orders, whereas support is lower at lag 1 on the same firms, so the
difference in support between the full lag-1 and lag-2 samples cannot be attributed to sample
composition alone.
\end{table}

\clearpage
\section{Sensitivity of the TEP results to the near-collinear innovation pair}
\label{app:tepsens}
The $14$-variable TEP specification of \MT{Section~5 of the main text}{Section~\ref{sec:data}} contains a near-collinear pair of
innovations: the residual correlation between $\texttt{xmeas\_10}$ (purge rate) and $\texttt{xmv\_6}$ (purge
valve) at the selected VAR lag is $-0.999$, which the input redundancy filter, applied to the differenced
series, does not detect. Two one-variable ablations re-estimate the pipeline of the main text on the stored
processed input with $\texttt{xmv\_6}$ removed (A) or $\texttt{xmeas\_10}$ removed (B): VAR lag fixed at the
baseline's selected value $6$, ICA seed $0$, the same $500$ moving-block resample index sets, $999$ block
permutations for the joint dHSIC, $499$ for the pairwise statistics, and $q=0.95$ for the tails. The design,
thresholds and comparisons were fixed after the initial conditioning and VAR checks and before the
full-pipeline ablations. Re-selecting the lag by the same AIC rule returned lag $6$ in both systems. The
baseline itself was reproduced from the stored input before the ablations were run. Table~\ref{tab:tepsens}
reports the comparison.

\begin{table}[!htbp]
\caption{\label{tab:tepsens}TEP: baseline and the two one-variable ablations.}
\begin{center}
\small
\begin{tabular*}{\textwidth}{@{}l@{\extracolsep{\fill}}ccc@{}}
\hline\hline
 & \shortstack[c]{Baseline\\($14$ variables)} & \shortstack[c]{Ablation A\\($\texttt{xmv\_6}$ removed)} & \shortstack[c]{Ablation B\\($\texttt{xmeas\_10}$ removed)} \\
\hline
\multicolumn{4}{@{}l}{\textit{Panel A: conditioning of the VAR(6) residual correlation matrix}} \\
Smallest eigenvalue $\lambda_{\min}$ & 0.000706 & 0.6046 & 0.6046 \\
Condition number $\kappa_2$ & 2911 & 2.3 & 2.3 \\
Largest absolute residual correlation & 0.999 & 0.202 & 0.199 \\
\multicolumn{4}{@{}l}{\textit{Panel B: recovered instantaneous coefficients (point fit, ICA seed 0)}} \\
$\texttt{xmeas\_10}\to\texttt{xmeas\_36}$ & 3.311 & absent & -- \\
$\texttt{xmv\_6}\to\texttt{xmeas\_36}$ & 3.278 & -- & absent \\
$\texttt{xmeas\_10}\to\texttt{xmv\_6}$ & $-0.999$ & -- & -- \\
$\texttt{xmeas\_29}\to\texttt{xmeas\_35}$ & $-0.201$ & $-0.189$ & $-0.205$ \\
$\texttt{xmeas\_27}\to\texttt{xmeas\_36}$ & $-0.197$ & $-0.210$ & $-0.207$ \\
$\texttt{xmeas\_35}\to\texttt{xmeas\_36}$ & $0.177$ & $0.183$ & $0.182$ \\
Largest absolute coefficient & 3.311 & 0.210 & 0.207 \\
Number of recovered edges & 10 & 6 & 7 \\
\multicolumn{4}{@{}l}{\textit{Panel C: bootstrap stability of the three edges that pass the persistence rule}} \\
$\texttt{xmeas\_29}\to\texttt{xmeas\_35}$ & 0.796 & 0.752 & 0.760 \\
$\texttt{xmeas\_27}\to\texttt{xmeas\_36}$ & 0.770 & 0.786 & 0.754 \\
$\texttt{xmeas\_35}\to\texttt{xmeas\_36}$ & 0.722 & 0.746 & 0.744 \\
\multicolumn{4}{@{}l}{\textit{Panel D: independence diagnostics on the recovered shocks}} \\
Joint dHSIC statistic & 0.00490 & 0.00518 & 0.00523 \\
Exceedances of $999$ permutations ($p$) & 0 (0.001) & 0 (0.001) & 0 (0.001) \\
Mean off-diagonal recovered-shock MI & 0.037 & 0.037 & 0.038 \\
Pairwise rejections at $10\%$ & 33 of 91 & 30 of 78 & 30 of 78 \\
Max tail $\widehat\chi$ (upper/lower) & 0.243/0.284 & 0.243/0.325 & 0.243/0.325 \\
Block-to-free null SD ratio & 29.53 & 24.67 & 18.55 \\
\hline\hline
\end{tabular*}
\end{center}
\footnotesize
\renewcommand{\baselineskip}{11pt}
\textbf{Note:} ``--'' = variable not in the system; ``absent'' = edge not retained in the point fit. The sign
consistency of the three edges in Panel~C is $1.00$ in every system; the paired Monte Carlo standard errors of
the stability differences between an ablation and the baseline, computed over the common resample index sets, are
$0.016$--$0.021$. Pairwise rejections are counted on the $78$ pairs among each system's $13$ variables; on
those same pairs the baseline rejects $26$ (the pairs of A) and $33$ (the pairs of B). In B the series
held fixed in the block permutation is $\texttt{xmeas\_23}$ rather than $\texttt{xmeas\_10}$.
\end{table}

Across ICA seeds $0$--$4$ (point fits at lag $6$; persistence and independence diagnostics were not recomputed for
each initialisation), the $14$-variable specification keeps $\texttt{xmeas\_10}\to\texttt{xmeas\_36}$ as the
leading directed edge in all five fits, but its coefficient ranges from $3.31$ to $4.31$, the
edge set at $|\hat b|\ge0.02$ overlaps with the seed-0 set by a Jaccard index of $0.21$--$0.90$, and one
initialisation returns two further coefficients above $2.7$ into $\texttt{xmeas\_29}$ together with the two
edges $\texttt{xmeas\_29}\to\texttt{xmeas\_35}$ and $\texttt{xmeas\_27}\to\texttt{xmeas\_36}$ in the opposite
orientation. In A and B, the largest absolute coefficient remains close to $0.21$ across the five fits. The
leading edge is $\texttt{xmeas\_27}\to\texttt{xmeas\_36}$ in four of them and $\texttt{xmeas\_29}\to\texttt{xmeas\_35}$,
of near-equal magnitude, in the fifth, and $\texttt{xmeas\_35}\to\texttt{xmeas\_36}$ appears with the opposite
orientation in one fit; four edges are present at $|\hat b|\ge0.02$ in all five fits of both systems. Seed $0$
remains the specification reported in the main text.

\end{document}